%% file: main.tex
\documentclass[]{eptcs}
\pdfoutput=1

\input{article_header.tex}

\input{tikzit_header2.tex}
\input{tikzit_generic.tex}

\input{tikzit_ct.tex}
\input{tikzit_GHZ-W.tex}

\title{A new Holant dichotomy inspired by quantum computation}
\author{Miriam Backens
\institute{School of Mathematics, University of Bristol, UK}
\email{m.backens@bristol.ac.uk}
}
\def\titlerunning{A new Holant dichotomy inspired by quantum computation}
\def\authorrunning{Miriam Backens}

\begin{document}

\maketitle

\begin{abstract}
 Holant problems are a framework for the analysis of counting complexity problems on graphs.
 This framework is simultaneously general enough to encompass many other counting problems on graphs and specific enough to allow the derivation of dichotomy results, partitioning all problem instances into those which can be solved in polynomial time and those which are \sP-hard.
 The Holant framework is based on the theory of holographic algorithms, which was originally inspired by concepts from quantum computation, but this connection appears not to have been explored before.
 
 Here, we employ quantum information theory to explain existing results in a concise way and to derive a dichotomy for a new family of problems, which we call \textsc{Holant}$^+$.
 This family sits in between the known families of \textsc{Holant}$^*$, for which a full dichotomy is known, and \textsc{Holant}$^c$, for which only a restricted dichotomy is known.
 Using knowledge from entanglement theory -- both previously existing work and new results of our own -- we prove a full dichotomy theorem for \textsc{Holant}$^+$, which is very similar to the restricted \textsc{Holant}$^c$ dichotomy.
 Other than the dichotomy for \#\textsc{R$_3$-CSP}, ours is the first Holant dichotomy in which the allowed functions are not restricted and in which only a finite number of functions are assumed to be freely available.
\end{abstract}

\input{intro.tex}

\input{Holant.tex}
\input{existing_work.tex}
\input{quantum_states.tex}
\input{ternary.tex}
\input{symmetrising.tex}
\input{binary.tex}

\input{hardness.tex}
\input{conclusion.tex}

\bibliographystyle{eptcs}
\bibliography{refs}

\appendix
\input{appendix.tex}

\end{document}

%% file: article_header.tex
\usepackage{amsmath} 
\usepackage{amsthm} 
\usepackage{amssymb}	
\usepackage{bold-extra} 
\usepackage{hyperref}
\usepackage{stmaryrd} 
\usepackage{gensymb} 
\usepackage{enumitem} 

\newtheorem{thm}{Theorem}
\newtheorem*{thm*}{Theorem}
\newtheorem{prop}[thm]{Proposition}
\newtheorem*{prop*}{Proposition}
\newtheorem{lem}[thm]{Lemma}
\newtheorem{cor}[thm]{Corollary}
\theoremstyle{definition}
\newtheorem{dfn}{Definition}
\newtheorem*{ex}{Example}

\newcommand{\avg}[1]{\left\langle #1 \right\rangle}
\newcommand{\CC}{\mathbb{C}}

\newcommand{\ZZ}{\mathbb{Z}}

\newcommand{\ket}[1]{\left| #1 \right>} 
\newcommand{\bra}[1]{\left< #1 \right|} 
\newcommand{\braket}[2]{\left< #1 \middle| #2 \right>} 
 
\renewcommand{\t}[1]{\ensuremath{^{\otimes #1}}}

\DeclareMathOperator{\spans}{span}

\newcommand{\fW}{\mathfrak{W}}

\newcommand{\cA}{\mathcal{A}}

\newcommand{\cE}{\mathcal{E}}
\newcommand{\cF}{\mathcal{F}}
\newcommand{\cG}{\mathcal{G}}
\newcommand{\cM}{\mathcal{M}}
\newcommand{\cT}{\mathcal{T}}
\newcommand{\cU}{\mathcal{U}}

\DeclareMathOperator{\Holant}{Holant}
\newcommand{\Holp}[2][]{\textsc{Holant}^{ #1 }\left( #2 \right)}
\newcommand{\csp}{\#\textsc{CSP}}

\newcommand{\FP}{\textsf{FP}}
\newcommand{\sP}{\#\textsf{P}}

\newcommand{\NP}{\textsf{NP}}

\newcommand{\etal}[0]{\emph{et al.}}

%% file: tikzit_header2.tex
\usepackage[svgnames]{xcolor} 

\usepackage{tikz}
\usetikzlibrary{shapes.geometric} 
\usetikzlibrary{shapes.misc} 
\usetikzlibrary{shapes.symbols} 
\usetikzlibrary{decorations.pathreplacing} 
\usetikzlibrary{decorations.markings} 
\usetikzlibrary{arrows} 

\pgfdeclarelayer{edgelayer}
\pgfdeclarelayer{nodelayer}
\pgfsetlayers{edgelayer,nodelayer,main}

\tikzstyle{none}=[inner sep=0pt]

\tikzstyle{every picture}=[baseline=-0.1cm,scale=0.5]

%% file: tikzit_generic.tex
\tikzstyle{unitary}=[rectangle,fill=white,draw=black,minimum height=14pt,minimum width=14pt,inner sep=0pt]
\tikzstyle{wideunitary}=[rectangle,fill=white,draw=black,minimum height=14pt,minimum width=24pt,inner sep=0pt]
\tikzstyle{graph}=[ellipse,fill=White,draw=Black,minimum height=0.5cm,minimum width=3cm,inner sep=0pt]
\tikzstyle{biggraph}=[ellipse,fill=White,draw=Black,minimum height=0.75cm,minimum width=6.5cm,inner sep=0pt]
\tikzstyle{scalardiamond}=[diamond,draw=black,fill=white,inner sep=1pt,minimum size=0.4cm]
\tikzstyle{dtr}=[regular polygon,regular polygon sides=3,shape border rotate=180,fill=white,draw=black,minimum size=22pt,inner sep=0pt]
\tikzstyle{utr}=[regular polygon,regular polygon sides=3,fill=white,draw=black,minimum size=22pt,inner sep=0pt]
\tikzstyle{dtrwide}=[isosceles triangle,isosceles triangle stretches,shape border rotate=270,fill=white,draw=black,minimum height=22pt,minimum width=1.5cm,inner sep=0pt]
\tikzstyle{bigcloud}=[cloud,fill=white,draw=black]
\tikzstyle{diagram}=[rectangle,fill=white,draw=black,minimum height=0.5cm,minimum width=1cm,inner sep=0pt]



%% file: tikzit_ct.tex
\tikzstyle{map}=[trapezium,trapezium left angle=90,trapezium right angle=120,fill=White,draw=Black,inner sep=0pt, minimum height=0.5cm]
\tikzstyle{map_dagger}=[trapezium,trapezium left angle=90,trapezium right angle=60,fill=White,draw=Black,inner sep=0pt, minimum height=0.5cm]
\tikzstyle{map_transpose}=[trapezium,trapezium left angle=120,trapezium right angle=90,fill=White,draw=Black,inner sep=0pt, minimum height=0.5cm]
\tikzstyle{transposemap}=[trapezium,trapezium left angle=120,trapezium right angle=90,fill=White,draw=Black,inner sep=0pt, minimum height=0.5cm] 

%% file: tikzit_GHZ-W.tex
\tikzstyle{solidn}=[circle,fill=Black,draw=Black,line width=0.8 pt,minimum size=5pt,inner sep=0pt]
\tikzstyle{hollown}=[circle,fill=White,draw=Black,line width=0.5 pt,minimum size=5pt,inner sep=0pt]
\tikzstyle{greyn}=[circle,fill=Grey,draw=Black,line width=0.5 pt,minimum size=5pt,inner sep=0pt]


%% file: intro.tex
\section{Introduction}

Quantum computation provided the inspiration for holographic algorithms \cite{valiant_holographic_2008}, which in turn inspired the Holant framework \cite{cai_holant_2009}.
While Holant problems are an area of active research, so far there appear to have been no attempts to apply knowledge from quantum information theory or quantum computation to their analysis.
Yet, as we show in the following, quantum information theory, and particularly the theory of quantum entanglement, offer promising new avenues of research into Holant problems.

The Holant framework encompasses a wide range of counting complexity problems on graphs, parameterised by sets of functions $\cF$.
A signature grid is a mathematical object constructed by assigning a complex-valued Boolean function $f_v\in\cF$ to each vertex $v$ in a finite graph in such a way that each edge incident on the vertex corresponds to an input of the function.
The signature grid is then assigned a complex number, called the \emph{Holant}, defined as:
\begin{equation}
 \sum_{\sigma:E\to\{0,1\}} \prod_{v\in V} f_v(\sigma|_{E(v)})
\end{equation}
where $E$ is the set of edges of the graph, $V$ the set of vertices, $\sigma$ an assignment of Boolean values to each edge, and $\sigma|_{E(v)}$ the restriction of $\sigma$ to the edges incident on $v$.
The associated counting problem $\Holp{\cF}$ is the following: given a signature grid with functions taken from $\cF$, find the value of the Holant \cite{cai_holant_2009}.

From a quantum theory perspective, a signature grid can be thought of as a tensor network, where each function is considered to be a tensor with one index for each input.
Then $\Holp{\cF}$ is the problem of evaluating the contraction of that tensor network.

Counting complexity problems that can be expressed in the Holant framework include the problem of counting matchings or perfect matchings, counting vertex covers \cite{cai_holant_2009}, or counting Eulerian orientations \cite{huang_dichotomy_2016}.
The Holant framework also encompasses other counting complexity frameworks like counting constraint satisfaction problems (\csp) or the counting complexity version of the graph homomorphism problem \cite{cai_holant_2009}.

On the other hand, the Holant framework is specific enough to allow the derivation of dichotomy theorems, which state that for function sets $\cF$ within certain classes, the Holant problem can either be solved in polynomial time or it is \sP-hard.
Such a dichotomy is not expected to hold for general counting complexity problems \cite{cai_dichotomy_2011}. 
One class for which a dichotomy has been derived is that of \textsc{Holant}$^*$ problems, which only considers function sets containing all unary functions \cite{cai_dichotomy_2011}.
We say that in \textsc{Holant}$^*$ problems, all unary functions are \emph{freely available}.
Another class, denoted \csp, involves only function sets containing equality functions of arbitrary arity, i.e.\ where equality functions are freely available \cite{cai_holant_2009,cai_holant_2012}.
A third class is symmetric \textsc{Holant}$^c$, where only two unary functions are freely available: the ones pinning edges to values 0 or 1, respectively.
Additionally, all functions are required to be \emph{symmetric}, meaning their value depends only on the Hamming weight of the input \cite{cai_holant_2012}.
Additional dichotomies exist for the plain \textsc{Holant} problem with no freely available functions, but these, too, restrict the functions sets considered: either to symmetric functions \cite{cai_complete_2013} or to real-valued functions only \cite{lin_complexity_2016}.
A full dichotomy for all Holant problems, as well as a dichotomy for \textsc{Holant}$^c$ without the symmetry restriction, have so far remained elusive.
While techniques derived from the idea of holographic algorithms play an important role in many of the dichotomy proofs, none of the existing results use the connection to quantum computation.

We on the other hand take inspiration from the origin of Holant problems in ideas taken from quantum computation in order to make a step towards a general dichotomy for \textsc{Holant}$^c$.

First, we analyse existing dichotomies in quantum terms and find that many of the polynomial-time solvable classes have very natural descriptions in quantum theory.
In particular, the \textsc{Holant}$^*$ dichotomy \cite{cai_dichotomy_2011} can be described in terms of the types of quantum entanglement present in the allowed functions.
Entanglement is a core concept in quantum theory: a quantum state of multiple systems is entangled if it cannot be written as a tensor product of states of subsystems.
For states of more than two systems, there are different types of entanglement which can be used for different information-theoretic tasks \cite{nielsen_quantum_2010} -- the classification of those entanglement types is an area of ongoing research \cite{dur_three_2000,verstraete_four_2002,lamata_inductive_2006,lamata_inductive_2007,backens_inductive_2016} (cf.\ Sections \ref{s:entanglement}--\ref{s:inductive_classification}).

We additionally find that the tractable class of \emph{affine functions} arising in the dichotomies for \csp{} and symmetric \textsc{Holant}$^c$ \cite{cai_holant_2009,cai_holant_2012} (see also Section \ref{s:csp}) is well-known in quantum information theory as \emph{stabilizer states} \cite{gottesman_heisenberg_1998}.

Motivated by this, we define a new class of Holant problems which we call \textsc{Holant}$^+$.
This class encompasses problems where four specific unary functions are freely available, including the two that are available in \textsc{Holant}$^c$ (see Section \ref{s:Holant_plus}).
In this way, \textsc{Holant}$^+$ fits between \textsc{Holant}$^*$, for which there is a general dichotomy, and \textsc{Holant}$^c$, for which there is no general dichotomy.
We pick these four unary functions because of their special role in quantum information theory and also because they enable us to use a known result from entanglement theory about producing two-system entangled states from many-system ones via projections \cite{popescu_generic_1992,gachechiladze_addendum_2016}: this corresponds to the ability to produce non-degenerate binary functions via gadgets.
In fact, we prove an extension of that result, enabling the construction of three-qubit entangled states, or equivalently ternary functions (cf.\ Section \ref{s:proof:three-qubit-entanglement}).

Using this, we derive our dichotomy theorem for \textsc{Holant}$^+$. The tractable classes in this dichotomy are very similar to those of the dichotomy for symmetric \textsc{Holant}$^c$ \cite{cai_holant_2009} (see also Section \ref{s:other_dichotomies}).
The one exception is that a family of several tractable sets for \textsc{Holant}$^c$, which are related to the set of affine functions, reduces simply to subsets of the affine functions in the \textsc{Holant}$^+$ case.

Our dichotomy is the first full Holant dichotomy -- i.e.\ a dichotomy which does not restrict the type of functions involved -- where only a finite number of functions is freely available, except for the dichotomy for \#\textsc{R$_3$-CSP} \cite{cai_complexity_2014}.
Furthermore, given that no dichotomy is known yet for \textsc{Holant}$^c$ with not necessarily symmetric functions, our result represents a step towards such a general \textsc{Holant}$^c$ dichotomy, as well as a general dichotomy for all Holant problems.

In the following, we introduce the Holant problem and associated notions in more detail in Section \ref{s:Holant_problem}.
In Section \ref{s:existing_results}, we recap the relevant existing dichotomies and results.
Next, we introduce the quantum perspective on Holant problems together with important notions from entanglement theory, this is in Section \ref{s:quantum_states}.
We define and motivate the new family of Holant problems called \textsc{Holant}$^+$ and prove the dichotomy theorem in Section \ref{s:Holant_plus}.

%% file: Holant.tex
\section{Holant problems}
\label{s:Holant_problem}

Holant problems are a framework for counting complexity problems on graphs, introduced by Cai \etal{} \cite{cai_holant_2009}, and based on the theory of holographic algorithms developed by Valiant \cite{valiant_holographic_2008}.

Let $\cF$ be a set of complex-valued functions with Boolean inputs, also called \emph{signatures}, and let $G=(V,E)$ be an undirected graph with vertices $V$ and edges $E$.
Throughout, graphs are allowed to have parallel edges (i.e.\ $E$ is a multi-set) and self-loops.
A \emph{signature grid} is a tuple $\Omega=(G,\cF,\pi)$ where $G$ is an undirected graph, $\cF$ is a set of functions, and $\pi$ is a function that assigns to each $n$-ary vertex $v\in V$ a function $f_v:\{0,1\}^n\to\CC$ in $\cF$ and also specifies which edge corresponds to which input.

A complex value called the \emph{Holant} can be associated with each signature grid:
\begin{equation}
 \Holant_\Omega = \sum_{\sigma:E\to\{0,1\}} \prod_{v\in V} f_v(\sigma |_{E(v)} ),
\end{equation}
where $\sigma |_{E(v)}$ denotes the restriction of $\sigma$ to the edges incident on $v$.

\begin{dfn}
 The Holant problem for a set of signatures $\cF$, denoted by \textsc{Holant}$(\cF)$, is defined as follows:
 \begin{description}[align=right, labelwidth=2cm]
  \item[Input] a signature grid $\Omega=(G,\cF,\pi)$ over the signature set $\cF$,
  \item[Output] $\Holant_\Omega$.
 \end{description}
\end{dfn}

The Holant problem is general enough to encode a wide range of counting problems defined on graphs, e.g. counting (not necessarily) perfect matchings, counting vertex covers, counting graph homomorphisms, or the counting constraint satisfaction problem \cite{cai_holant_2009,cai_graph_2010}.
Simultaneously, it is specific enough to allow the derivation of dichotomy theorems: for all families of sets of signatures analysed so far, the complexity classification takes the form:
\begin{quotation}
 \noindent For any signature set $\cF$ in the family, either \textsc{Holant}$(\cF)$ can be solved in polynomial time or \textsc{Holant}$(\cF)$ is \sP-hard.
\end{quotation}
By an analogue of Ladner's Theorem about \NP-intermediate problems \cite{ladner_structure_1975}, such a dichotomy is not expected to hold for general counting problems \cite{cai_dichotomy_2011}.

A \emph{symmetric signature} is a function that depends only on the Hamming weight of the input.
In other words, the value of this function does not change under interchange of any two inputs.
Let $f:\{0,1\}^n\to\CC$ be symmetric.
Then $f$ is often written as:
\begin{equation}
 f = [f_0,f_1,\ldots,f_n],
\end{equation}
where $f_k$ is the value $f$ takes on inputs of Hamming weight $k$ for $k=0,\ldots,n$.
Many complexity results for the Holant problem are specifically about symmetric signatures, e.g. \cite{cai_holant_2012,cai_dichotomy_2011}.

A signature is called \emph{degenerate} if it is a product of unary signatures.
Any signature that cannot be expressed as a product of unary signatures is called \emph{non-degenerate}.
For example, the binary signature $f$ satisfying:
\begin{equation}
 f(x,y) = \begin{cases} 1&\text{if } x=y=0 \\ 0&\text{otherwise} \end{cases}
\end{equation}
is degenerate: it can be written as $f(x,y)=g(x)g(y)$, where:
\begin{equation}
 g(x) = \begin{cases} 1&\text{if } x=0 \\ 0&\text{if } x=1. \end{cases}
\end{equation}
The binary equality signature $(=_2) = [1,0,1]$ on the other hand is non-degenerate: it is impossible to find two unary functions whose product is $=_2$.

Given a bipartite graph, we can define a \emph{bipartite signature grid} by specifying two signature sets $\cF$ and $\cG$ and assigning vertices from the first (second) partition signatures from $\cF$ ($\cG$).
A bipartite signature grid is denoted by a tuple $(G,\cF\mid\cG,\pi)$.
The corresponding \emph{bipartite Holant problem} is \textsc{Holant}$(\cF\mid\cG)$.

\subsection{Signature grids in terms of vectors}
\label{s:vector_perspective}

As noted in \cite{cai_valiants_2006}, any signature $f:\{0,1\}^n\to\CC$ can be considered as a complex vector of $2^n$ components indexed by $\{0,1\}^n$.

Let $\{\ket{x}\}_{x\in\{0,1\}^n}$ be an orthonormal basis for $\CC^{2^n}$.\footnote{In using this notation for vectors, called \emph{Dirac notation} and common in the field of quantum computing and quantum information theory, we anticipate the interpretation of the vectors associated to signatures as quantum states, cf.\ Section \ref{s:quantum_states}.}
The vector corresponding to the signature $f$ is then denoted by:
\begin{equation}
 \ket{f} = \sum_{x\in\{0,1\}^n} f(x)\ket{x}.
\end{equation}

Suppose $\Omega=(G,\cF\mid\cG,\pi)$ is a bipartite signature grid, where $G=(V,W,E)$ has vertex partitions $V$ and $W$.
Then the Holant for $\Omega$ can be written as:
\begin{equation}\label{eq:bipartite_Holant_vectors}
 \Holant_\Omega = \left(\bigotimes_{w\in W} \left(\ket{g_w}\right)^T\right) \left(\bigotimes_{v\in V} \ket{f_v}\right) = \left(\bigotimes_{v\in V} \left(\ket{f_v}\right)^T\right) \left(\bigotimes_{w\in W} \ket{g_w}\right),
\end{equation}
where the tensor products are assumed to be ordered such that, everywhere, the two systems associated with the same edge meet.

\subsection{Reductions}

Given two counting problems $A$ and $B$, the expression $A \leq_T B$ denotes that there exists a polynomial time reduction from problem $A$ to problem $B$, i.e.\ the complexity of $A$ is at most that of $B$.
If $A\leq_T B$ and $B\leq_T A$, we write $A\equiv_T B$.

\subsubsection{Holographic reductions}

The most important type of reduction -- and the source of the name `Holant problem' -- are holographic reductions.
Let $M$ be a 2 by 2 complex matrix.
Then, for any $f:\{0,1\}^n\to\CC$, we write $M\circ f$ for the function corresponding to the vector $M\t{n}\ket{f}$.
Furthermore, for any signature set $\cF$, we write:
\begin{equation}
 M\circ\cF = \{ M\circ f : f\in\cF \}.
\end{equation}

\begin{thm}[Valiant's Holant Theorem, \cite{valiant_holographic_2008}]\label{thm:Valiant_Holant}
 Suppose $\cF$ and $\cG$ are two sets of signatures, $M$ an invertible 2 by 2 complex matrix, and $\Omega=(G,\cF\mid\cG, \pi)$ a signature grid.
 Let $\Omega'=(G,M\circ\cF\mid (M^{-1})^T\circ\cG, \pi')$ be the signature grid resulting from $\Omega$ by replacing each $f_v$ or $g_w$ by $M\circ f_v$ or $(M^{-1})^T\circ g_w$, respectively. Then:
 \begin{equation}
  \Holant_\Omega = \Holant_{\Omega'}.
 \end{equation}
\end{thm}

For completeness and to illustrate the use of the vector notation, we give a proof.

\begin{proof}
 Use the definition of the Holant in terms of vectors, \eqref{eq:bipartite_Holant_vectors}. Then:
 \begin{align*}
  \Holant_{\Omega'} &= \left(\bigotimes_{w\in W} \left(\left(\left(M^{-1}\right)^T\right)\t{arity(g_w)}\ket{g_w}\right)^T\right) \left(\bigotimes_{v\in V} M\t{arity(f_v)}\ket{f_v}\right) \\
  &= \left(\bigotimes_{w\in W} \left(\ket{g_w}\right)^T \left(M^{-1}\right)\t{arity(g_w)} \right) \left(\bigotimes_{v\in V} M\t{arity(f_v)}\ket{f_v}\right),
  \intertext{so on each system in the tensor product, a copy of $M^{-1}$ meets a copy of $M$. Hence:}
  \Holant_{\Omega'} &= \left(\bigotimes_{w\in W} \left(\ket{g_w}\right)^T \right) \left(\bigotimes_{v\in V} \ket{f_v}\right) = \Holant_\Omega. \qedhere
 \end{align*}
\end{proof}

\begin{cor}[\cite{valiant_holographic_2008}]
 Suppose $\cF$ is a set of signatures and $O$ is a 2 by 2 complex orthogonal matrix, i.e.\ $O^TO=OO^T=I$. Let $\Omega = (G,\cF,\pi)$ be a signature grid and let $\Omega' = (G,O\circ\cF,\pi')$ be the signature grid that results from $\Omega$ by replacing $f_v$ with $O\circ f_v$. Then:
 \begin{align}
  \Holant_\Omega = \Holant_{\Omega'}.
 \end{align}
\end{cor}
\begin{proof}
 Any signature grid can be made bipartite by adding a vertex in the middle of each edge and assigning to each new vertex the binary equality signature: i.e.\ given a signature grid $\Omega = (G,\cF,\pi)$, we can construct a bipartite signature grid:
 \begin{equation}
  \Omega'' = (G', \cF\mid\{=_2\}, \pi'')
 \end{equation}
 where $G'$ results from $G$ by adding an extra vertex in the middle of each edge and $\pi''$ is the function that is $\pi$ on the original vertices and assigns $=_2$ to the new ones.
 Then, by construction, $\Holant_{\Omega}=\Holant_{\Omega''}$.
 
 Now we can apply Theorem \ref{thm:Valiant_Holant} to $\Omega''$ to construct a new bipartite signature grid $\Omega'''=(G',O\circ\cF\mid\{=_2\},\pi'')$ with $\Holant_{\Omega''}=\Holant_{\Omega'''}$, noting that:
 \begin{equation}
  (\ket{=_2})^T \left(O^{-1}\right)\t{2} = (\ket{=_2})^T
 \end{equation}
 for any 2 by 2 complex orthogonal matrix $O$.
 
 The corollary then follows by transforming $\Omega'''$ back into a non-bipartite signature grid by removing the equality vertices and merging the edges incident on them.
\end{proof}

As a result of Valiant's Holant theorem, we thus have:
\begin{equation}
 \Holp{\cF\mid\cG} \equiv_T \Holp{M\circ\cF\mid (M^{-1})^T\circ\cG}
\end{equation}
for any invertible 2 by 2 complex matrix $M$, and:
\begin{equation}
 \Holp{\cF} \equiv_T \Holp{O\circ\cF}
\end{equation}
for any orthogonal 2 by 2 complex matrix $O$.
The process of going from a signature set $\cF$ to a set $M\circ\cF$ is a holographic reduction.

\subsubsection{Gadgets and polynomial interpolation}

A \emph{gadget} over a signature set $\cF$ (also called $\cF$-gate) is a fragment of a signature grid with some `dangling' edges.
Any gadget can be assigned an effective signature ${g}$.
If ${g}$ is the effective signature of some gadget over $\cF$, ${g}$ is said to be \emph{realisable over $\cF$}.

\begin{lem}[\cite{cai_dichotomy_2011}]
 Suppose $\cF$ is some signature set and ${g}$ is realisable over $\cF$. Then:
 \begin{equation}\label{eq:F_cup_psi}
  \Holp{\cF\cup\{{g}\}} \equiv_T \Holp{\cF}.
 \end{equation}
\end{lem}

As multiplying one or more signatures by a constant does not change the complexity of a Holant instance, we consider a signature realisable if some multiple of it by a non-zero complex constant is realisable.

Following \cite{lin_complexity_2016}, we define for any signature set $\cF$:
\begin{equation}
 S(\cF) = \{ {g} \mid {g} \text{ is realisable over } \cF \}.
\end{equation}
Then:
\begin{equation}
 \Holp{S(\cF)} \equiv_T \Holp{\cF}.
\end{equation}

The concept of a gadget can be extended to bipartite signature grids.
A left-side gadget over $\cF\mid\cG$ is a fragment of a bipartite signature grid where all dangling edges are connected to vertices from the right partition, i.e.\ the gadget can be used as if it was in $\cF$.
A right-side gadget can be defined analogously.

If ${g}\notin S(\cF)$, in certain cases it is nevertheless possible to show a result like \eqref{eq:F_cup_psi} by analysing a family of signature grids that differ in specific ways.
This process is called \emph{polynomial interpolation} and will not be used here, though it is a crucial ingredient in some of the results we build upon.
The interested reader can find a discussion of polynomial interpolation in \cite{cai_holant_2009}.

%% file: existing_work.tex
\section{Existing results about the Holant problem}
\label{s:existing_results}

We now introduce the existing families of Holant problems and the associated dichotomy results.

Gadget constructions, which are at the heart of many reductions, are easier the more freely available signatures there are.
As a result, several families of Holant problems have been defined, in which certain sets of signatures are freely available (and can thus be used in gadget constructions and polynomial interpolation).
These families are:
\begin{itemize}
 \item Complex-weighted Boolean \csp, a counting constraint satisfaction problem, which corresponds to a Holant problem in which equality functions of any arity are freely available. Formally:
  \begin{equation}
   \csp(\cF) = \Holp{\cF\cup\cG},
  \end{equation}
  where $\cG=\{ =_1, =_2, =_3, \ldots \}$ with $=_1$ being the function that is equal to 1 one both inputs \cite{cai_holant_2009,cai_holant_2012,cai_complexity_2014}.
 \item $\Holp[*]{\cF}$, the Holant problem in which all unary signatures are available for free, i.e.:
  \begin{equation}
   \Holp[*]{\cF} = \Holp{\cF\cup\cU},
  \end{equation}
  where $\cU$ is the set of all unary signatures \cite{cai_holant_2009,cai_dichotomy_2011}.
 \item $\Holp[c]{\cF}$, the Holant problem in which edges can be pinned to be 0 or 1, respectively. Formally:
  \begin{equation}
   \Holp[c]{\cF} = \Holp{\cF\cup\{\delta_0,\delta_1\}},
  \end{equation}
  where $\delta_0(0)=1$, $\delta_0(1)=0$, and the other way around for $\delta_1$ \cite{cai_holant_2009,cai_holant_2012}.
\end{itemize}

In addition, there is, of course, plain $\textsc{Holant}$: the full Holant problem with no freely-available signatures \cite{cai_complete_2013,lin_complexity_2016}.

\subsection{The \textsc{Holant}$^*$ dichotomy}
\label{s:Holant_star}

We first note the dichotomy for the Holant problem in which all unary signatures are freely available.

Given a bit string $x$, let $\bar{x}$ be its bit-wise complement. For a set of signatures $\cF$, denote by $\avg{\cF}$ the closure of $\cF$ under tensor products.
Furthermore, let:
\begin{itemize}
 \item $\cT$ be the set of all binary signatures,
 \item $\cE$ the set of signatures which are non-zero only on two inputs $x$ and $\bar{x}$, and
 \item $\cM$ the set of signatures which are non-zero only on inputs of Hamming weight at most 1.
\end{itemize}
Finally, define:\footnote{The matrix we denote by $K$ is usually denoted by $Z$ in the literature; we have changed the label to avoid confusion with the Pauli-$Z$ matrix commonly used in quantum theory.}
\begin{equation}
 K = \begin{pmatrix}1&1\\i&-i\end{pmatrix} \quad\text{and}\quad X = \begin{pmatrix}0&1\\1&0\end{pmatrix}.
\end{equation}
The matrix $K$ is useful because it satisfies $K^TK\doteq X$, where `$\doteq$' denotes equality up to non-zero scalar factor.
In fact, up to multiplication by a diagonal matrix or by $X$ itself, $K$ is the only solution to this equation; see Appendix \ref{a:ATA=X}.

\begin{thm}[\cite{cai_dichotomy_2011}]\label{thm:Holant-star}
 Let $\cF$ be any set of complex valued functions in Boolean variables. The problem $\Holp[*]{\cF}$ is polynomial time computable if:
 \begin{itemize}
  \item $\cF\subseteq\avg{\cT}$, or
  \item $\cF\subseteq\avg{O\circ\mathcal{E}}$, where $O$ is a complex orthogonal 2 by 2 matrix, or
  \item $\cF\subseteq\avg{K\circ\mathcal{E}}$, or
  \item $\cF\subseteq\avg{K\circ\mathcal{M}}$ or $\cF\subseteq\avg{KX\circ\mathcal{M}}$.
 \end{itemize}
 In all other cases, $\Holp[*]{\cF}$ is \sP-hard. The dichotomy is still valid even if the inputs are restricted to planar graphs.
\end{thm}

Note that the set $\avg{KX\circ\mathcal{E}}$ is equal to $\avg{K\circ\mathcal{E}}$ since $X\in\cE$, hence it does not need to be listed separately.
Furthermore, the set $\avg{\cE}$ itself is included in the second tractable case by taking $O$ to be the identity matrix.
Subsets of $\cM$ on the other hand become tractable only after a holographic transformation by $K$ or $KX$.

\subsection{The \csp{} dichotomy}
\label{s:csp}

The set of equality signatures has rather different properties to that of unary signatures, hence the \csp{} dichotomy is somewhat different to that for \textsc{Holant}$^*$.
In particular, a new family of tractable signatures arise, called `affine signatures'.

\begin{dfn}\label{dfn:affine_signature}
 A signature $f:\{0,1\}^n\to\CC$ is called \emph{affine} if it has the following form:
 \begin{equation}
  f(x) = c i^{l(x)} (-1)^{q(x)} \chi_{Ax=b},
 \end{equation}
 where $c\in\CC$, $l:\{0,1\}^n\to\ZZ_2$ is a linear function, $q:\{0,1\}^n\to\ZZ_2$ is a quadratic function, $A$ is a $m$ by $n$ matrix with Boolean entries for some $0\leq m\leq n$, $b\in\{0,1\}^m$, and $\chi$ is a 0-1 indicator function:
 \begin{equation}
  \chi_{Ax=b} = \begin{cases} 1 & \text{if } Ax=b\\ 0 & \text{otherwise.} \end{cases}
 \end{equation}
\end{dfn}
In the above, $i^2=-1$.
The set of $x$ such that $Ax=b$ form an affine space, hence the name for this class of signatures.
Note that all equality signatures are affine.

For the reader familiar with quantum information theory, the affine signatures correspond -- up to a scalar factor -- to stabilizer states (cf.\ Section \ref{s:existing_quantum}).

Let $\cA$ be the set of affine signatures, which is already closed under tensor products.
Then the dichotomy for \csp{} takes the following form.

\begin{thm}[\cite{cai_complexity_2014}]\label{thm:csp}
 Suppose $\cF$ is a class of functions mapping Boolean inputs to complex numbers. If $\cF\subseteq\mathcal{A}$ or $\cF\subseteq\avg{\mathcal{E}}$, then \csp($\cF$) is computable in polynomial time. Otherwise, \csp($\cF$) is \sP-hard.
\end{thm}

The same dichotomy also holds for \#\textsc{R$_3$-CSP}, which corresponds to the following bipartite Holant problem \cite{cai_complexity_2014}:
\begin{equation}
 \#\textsc{R$_3$-CSP}(\cF) = \Holp{\cF\mid\{=_1,=_2,=_3\}}.
\end{equation}
This dichotomy follows immediately from that for \csp{} if $\cF$ contains the binary (or indeed any non-unary) equality function, but it is non-trivial if $\cF$ does not contain any non-unary equality functions.

\subsection{Other dichotomies}
\label{s:other_dichotomies}

There is no full dichotomy for \textsc{Holant}$^c$ yet, though there is a dichotomy that applies to sets of symmetric signatures only.
This dichotomy for symmetric \textsc{Holant}$^c$ combines the tractable classes of \textsc{Holant}$^*$ and \csp{} (up to an additional holographic transformation in the latter case).

\begin{thm}[\cite{cai_holant_2012}]\label{thm:Holant-c}
 Let $\mathcal{F}$ be a set of complex symmetric signatures. $\Holp[c]{\cF}$ is \sP-hard unless $\cF$ satisfies one of the following conditions, in which case it is tractable:
 \begin{itemize}
  \item $\Holp[*]{\cF}$ is tractable (cf. Theorem \ref{thm:Holant-star}), or
  \item there exists a $T\in\mathcal{I}$ such that $\cF\subseteq T\circ\cA$, where:
   \begin{equation}
    \mathcal{I} = \left\{ T \,\middle|\, \left(T^{-1}\right)^T\circ \{ =_2, \delta_0, \delta_1 \} \subset \mathcal{A} \right\}.
   \end{equation}
 \end{itemize}
\end{thm}

In the case of \textsc{Holant} with no free signatures, there exist the following results:
\begin{itemize}
 \item a dichotomy for complex-valued symmetric signatures \cite{cai_complete_2013}, and
 \item a dichotomy for (not necessarily symmetric) signatures taking non-negative real values \cite{lin_complexity_2016}.
\end{itemize}
We shall not explore those in any detail here.

\subsection{Results about ternary symmetric signatures}
\label{s:results_ternary_symmetric}

There are some comprehensive results classifying the hardness of the Holant problem for bipartite signature sets of the form:
\begin{equation}
 \{[y_0,y_1,y_2]\} \mid \{[x_0,x_1,x_2,x_3]\},
\end{equation}
i.e.\ where one partition only contains binary vertices, the other only contains ternary ones, and all vertices of the same arity are assigned the same symmetric signature.
In this case, the curly braces around the signature sets are often dropped from the notation.
Furthermore, if the ternary signature is non-degenerate, it can always be mapped to $[1,0,0,1]$ or $[1,1,0,0]$ by a holographic transformation \cite{cai_holant_2012}.
If the ternary signature is degenerate, the problem is tractable by the first case of Theorem \ref{thm:Holant-star}.
It thus suffices to consider the cases:
\begin{equation}
 [y_0,y_1,y_2]\mid[1,0,0,1] \quad\text{and}\quad [y_0,y_1,y_2]\mid[1,1,0,0].
\end{equation}

For $[1,0,0,1]$, note that there are non-trivial holographic transformations leaving this signature invariant \cite{cai_holant_2012}.
In particular:
\begin{equation}
 \begin{pmatrix}1&0\\0&\omega\end{pmatrix}\circ [1,0,0,1] = [1,0,0,1],
\end{equation}
where $\omega$ is a third root of unity, i.e.\ $\omega^3=1$.
Thus, by Valiant's Holant Theorem:
\begin{equation}\label{eq:normalisation}
 \Holp{[y_0,y_1,y_2]\mid[1,0,0,1]} \equiv_T \Holp{[y_0, \omega y_1, \omega^2 y_2]\mid[1,0,0,1]}.
\end{equation}
This relationship can be used to reduce the number of symmetric binary signatures needing to be considered.
Following \cite{cai_holant_2012}, a signature of the form $[y_0,y_1,y_2]$ is called \emph{$\omega$-normalised}\footnote{We use the term $\omega$-normalisation to distinguish it from other notions of normalisation, e.g.\ ones relating to the norm of the vector associated with a signature.} if:
\begin{itemize}
 \item $y_0=0$, or
 \item there does not exist a primitive $(3t)$-th root of unity $\lambda$, where $gcd(t,3)=1$, such that $y_2=\lambda y_0$.
\end{itemize}
Similarly, a unary signature $[a,b]$ is $\omega$-normalised if:
\begin{itemize}
 \item $a=0$, or
 \item there does not exist a primitive $(3t)$-th root of unity $\lambda$, where $gcd(t,3)=1$, such that $b=\lambda a$.
\end{itemize}
If a binary signature is not $\omega$-normalised, it can be made so through application of a holographic transformation of the form given in \eqref{eq:normalisation}.
Unary signatures will only be required when the binary signature has the form $[0,y_1,0]$; in that case the binary signature is $\omega$-normalised and remains so under a holographic transformation that $\omega$-normalises the unary signature.

These definitions allow a characterisation of the Holant problem for bipartite signature grids, where there is a ternary equality signature on one partition and a non-degenerate symmetric binary signature on the other partition.

\begin{thm}[\cite{cai_holant_2012}]\label{thm:GHZ-state}
 Let $\mathcal{G}_1,\mathcal{G}_2$ be two sets of signatures and let $[y_0,y_1,y_2]$ be a $\omega$-normalised and non-degenerate signature.
 In the case of $y_0=y_2=0$, further assume that $\mathcal{G}_1$ contains a unary signature $[a,b]$ which is $\omega$-normalised and satisfies $ab\neq 0$.
 Then:
 \begin{equation}
  \Holp{\{[y_0,y_1,y_2]\}\cup\mathcal{G}_1 \mid \{[1,0,0,1]\}\cup\mathcal{G}_2} \equiv_T \csp(\{[y_0,y_1,y_2]\}\cup\mathcal{G}_1\cup\mathcal{G}_2).
 \end{equation}
 More specifically, $\Holp{\{[y_0,y_1,y_2]\}\cup\mathcal{G}_1 \mid \{[1,0,0,1]\}\cup\mathcal{G}_2}$ is \sP-hard unless:
 \begin{itemize}
  \item $\{[y_0,y_1,y_2]\}\cup\mathcal{G}_1\cup\mathcal{G}_2\subseteq\avg{\cE}$, or
  \item $\{[y_0,y_1,y_2]\}\cup\mathcal{G}_1\cup\mathcal{G}_2\subseteq\mathcal{A}$,
 \end{itemize}
 in which cases the problem is in \FP.
\end{thm}

Any signature reducible to $[1,1,0,0]$ via holographic transformations can be written as $[x_0,x_1,x_2,x_3]$ with \cite{cai_holant_2009}:
\begin{itemize}
 \item $x_k=Ak\alpha^{k-1}+B\alpha^k$, where $A\neq 0$, or
 \item $x_k=A(3-k)\alpha^{2-k}+B\alpha^{3-k}$, where $A\neq 0$.
\end{itemize}
For $\alpha=0$, $k\alpha^{k-1}$ is considered to be 1 if $k=1$ and zero for the other values of $k$.
The second case above is equivalent to the first one under a relabelling of the inputs $0\leftrightarrow 1$.
It thus suffices to consider the first case without loss of generality.
This yields the following result.

\begin{lem}[\cite{cai_holant_2012}]\label{lem:W-state}
 Let $x_k = Ak\alpha^{k-1} + B\alpha^k$, where $A\neq 0$ and $k=0,1,2,3$. $\Holp{[x_0,x_1,x_2,x_3]}$ is \sP-hard unless $\alpha=\pm i$, in which case the problem is in \FP.
\end{lem}

Note that the case $\alpha=\pm i$ corresponds exactly to $[x_0,x_1,x_2,x_3]\in K\circ\cM$ or $[x_0,x_1,x_2,x_3]\in KX\circ\cM$.

In the case where a binary symmetric signature is present, we have:

\begin{thm}[\cite{cai_holant_2012}]\label{thm:W-state}
 $\Holant([y_0,y_1,y_2]|[x_0,x_1,x_2,x_3])$ is \sP-hard unless $[x_0,x_1,x_2,x_3]$ and $[y_0,y_1,y_2]$ satisfy one of the following conditions, in which case the problem is in \FP:
 \begin{itemize}
  \item $[x_0,x_1,x_2,x_3]$ is degenerate, or
  \item there is a 2 by 2 matrix $M$ such that:
   \begin{itemize}
    \item $[x_0,x_1,x_2,x_3]=M\circ[1,0,0,1]$ and $(M^T)^{-1}\circ[y_0,y_1,y_2]$ is in $\mathcal{A}\cup\mathcal{P}$,
    \item $[x_0,x_1,x_2,x_3]=M\circ[1,1,0,0]$ and $(M^T)^{-1}\circ[y_0,y_1,y_2]$ is of the form $[0,*,*]$,
    \item $[x_0,x_1,x_2,x_3]=M\circ[0,0,1,1]$ and $(M^T)^{-1}\circ[y_0,y_1,y_2]$ is of the form $[*,*,0]$,
   \end{itemize}
 \end{itemize}
 with $*$ denoting an arbitrary complex number.
\end{thm}

Hence, to show that a Holant problem is \sP-hard, it suffices to construct a ternary symmetric non-degenerate signature that falls into one of the \sP-hard cases of one of the above theorems when combined with the binary equality signature; or a ternary and a non-trivial binary signature that do.
Yet even that gadget construction is not in general easy or necessarily even possible.
We therefore introduce a new family of Holant problems, motivated by results from quantum theory, and explicitly designed to enable the construction of ternary symmetric non-degenerate signatures.

%% file: quantum_states.tex
\section{The quantum state perspective on signature grids}
\label{s:quantum_states}

In Section \ref{s:vector_perspective}, we introduced the idea of considering signatures as complex vectors.
This perspective is useful for proving Valiant's Holant Theorem, which is at the heart of the theory of Holant problems.

But the vector notation has further advantages: it gives a connection to the theory of quantum computation.
In quantum computation and quantum information, the basic system of interest is a \emph{qubit} (quantum bit), which takes the place of the usual bit in standard computer science.
The state of a qubit is described by a vector\footnote{Strictly speaking, vectors only describe \emph{pure} quantum states: there are also \emph{mixed} states, which need to be described differently; but we do not consider those here.} in $\CC^2$.
The state of $n$ qubits is described by a vector in:
\begin{equation}
  \left(\CC^2\right)\t{n} := \underbrace{\CC^2 \otimes \CC^2 \otimes \ldots \otimes \CC^2}_{n\text{ copies}}.
\end{equation}
Now, $\left(\CC^2\right)\t{n}$ is isomorphic to $\CC^{2^n}$.
Thus, the vector associated with an $n$-ary signature can be considered to be a quantum state of $n$ qubits.
Quantum states are normally required to have norm 1, but for the methods used here, multiplication by a non-zero complex number does not make a difference, so we can work with states having arbitrary norms.

Let $\{\ket{0},\ket{1}\}$ be an orthonormal basis for $\CC^2$.
We call this the \emph{computational basis}.
The induced basis on $\left(\CC^2\right)\t{n}$ is labelled by $\{\ket{x}\}_{x\in\{0,1\}^n}$ as a short-hand, e.g.\ we write:
\begin{equation}
  \ket{00\ldots 0} := \ket{0}\otimes\ket{0}\otimes\ldots\otimes\ket{0}.
\end{equation}
This is just the same as the basis introduced in Section \ref{s:vector_perspective}.

In quantum theory, the Hermitian adjoint or conjugate transpose of a vector $\ket{\psi}$ is written as $\bra{\psi}$.
The inner product of a vector $\ket{\psi}$ with a vector $\ket{\phi}$ is then denoted $\braket{\psi}{\phi}$.
The Holant picture does not involve any complex conjugation.
Nevertheless, we will sometimes use the notation $\braket{\psi}{\phi}$, either if the components of $\ket{\psi}$ are real so that complex conjugation leaves it invariant, or if $\ket{\psi}$ is an arbitrary element of a set that is closed under complex conjugation.

If $\ket{\phi}$ is a single-qubit state and $\ket{\psi}$ an $n$-qubit state, the notation $\bra{\phi}_k\ket{\psi}$ for $1\leq k\leq n$, denotes the signature associated with the gadget in which a node with signature $\ket{\phi}$ is connected to the $k$-th output of $\ket{\psi}$. 

From now on, we will use standard Holant terminology and quantum terminology interchangeably, and sometimes mix the two.

\subsection{Holographic transformations and SLOCC}

Holographic transformations also have a natural interpretation in quantum information theory: going from an $n$-qubit state $\ket{f}$ to $M\t{n}\ket{f}$, where $M$ is some invertible 2 by 2 matrix, is a `stochastic local operation with classical communication' or SLOCC \cite{bennett_exact_2000,dur_three_2000}.
This term means that if $n$ people hold a qubit each and the $n$ qubits are in the joint state $\ket{f}$, there is a procedure for transforming the state $\ket{f}$ to $M\t{n}\ket{f}$ using only:
\begin{itemize}
 \item local operations, i.e.\ operations that can be applied on one qubit without needing access to the others, and
 \item classical communication, i.e.\ communication of non-quantum information,
\end{itemize}
which succeeds with non-zero probability.

SLOCC operations are slightly more general than holographic transformations, in that the former do not need to be symmetric under interchange of the qubits.
The most general SLOCC operation on an $n$-qubit state is given by:
\begin{equation}
  M_1\otimes M_2\otimes \ldots \otimes M_n,
\end{equation}
where $M_1,M_2,\ldots M_n$ are invertible complex 2 by 2 matrices \cite{dur_three_2000}.

\subsection{Entanglement and its classification}
\label{s:entanglement}

One major difference between quantum theory and preceding theories of physics (known as `classical physics') is the possibility of \emph{entanglement} in states of multiple systems.

\begin{dfn}
 A state of multiple systems is \emph{entangled} if it cannot be written as a tensor product of states of individual systems.
\end{dfn}

\begin{ex}
 In the case of two qubits:
 \begin{equation}
  \ket{00}+\ket{01}+\ket{10}+\ket{11}
 \end{equation}
 is a product state -- it can be written as $(\ket{0}+\ket{1})\otimes(\ket{0}+\ket{1})$.
 On the other hand, consider the state:
 \begin{equation}
  \ket{00}+\ket{11}.
 \end{equation}
 It is impossible to find single-qubit states $\ket{f},\ket{g}\in\CC^2$ such that $\ket{f}\otimes\ket{g} = \ket{00}+\ket{11}$.
 Thus, $\ket{00}+\ket{11}$ is entangled.
\end{ex}

Where a state involves more than two systems, it is possible for some of the systems to be entangled with each other and for other systems to be in a product state with respect to the former.
Even when all systems are entangled with each other, when there are more than two subsystems, there are different ways in which this can happen.
We sometimes use the term \emph{genuinely entangled state} to refer to a state in which no subsystem is in a product state with the others.
For example, $\ket{000}+\ket{111}$ is genuinely entangled but $\ket{0}\otimes(\ket{00}+\ket{11})$ is not.

Entanglement is an important resource in quantum computation, where it has been shown that quantum speedups are impossible without the presence of unboundedly growing amounts of entanglement \cite{jozsa_role_2003}.
Similarly, it is a resource in quantum information theory.
There, the standard set-up for quantum tasks involves several parties sharing an entangled state.
Each party may perform arbitrary (physically-allowed) operations locally on her subsystem but does not have access to the other subsystems.
Additionally, the parties can communicate over a classical (i.e.\ non-quantum) channel \cite{nielsen_quantum_2010}.
Examples of such protocols include quantum teleportation \cite{bennett_teleporting_1993} and quantum key distribution \cite{ekert_quantum_1991}.

It therefore makes sense to introduce the following equivalence relation on entangled states.

\begin{dfn}
 Two $n$-qubit states are \emph{equivalent under SLOCC} if one can be transformed into the other using SLOCC.
 More formally: suppose $\ket{f}$ and $\ket{g}$ are two $n$-qubit states. Then $\ket{f} \sim_{SLOCC} \ket{g}$ if and only if there exist invertible complex 2 by 2 matrices $M_1,M_2,\ldots M_n$ such that:
 \begin{equation}
  \left(M_1\otimes M_2\otimes \ldots \otimes M_n\right) \ket{f} = \ket{g}.
 \end{equation}
\end{dfn}
The equivalence classes of this relation are called \emph{entanglement classes} or \emph{SLOCC classes}.

For two qubits, there is only one class of entangled states, i.e.\ all entangled two-qubit states are equivalent to $\ket{00}+\ket{11}$ under SLOCC.
For three qubits, there are two classes of genuinely entangled states \cite{dur_three_2000}.
The first one is called GHZ class, these are states that are equivalent under SLOCC to the GHZ state:
\begin{equation}
 \ket{GHZ} := \ket{000}+\ket{111}.
\end{equation}
The second one is called W class and the standard representative is the W state:
\begin{equation}
 \ket{W} := \ket{001}+\ket{010}+\ket{100}.
\end{equation}

For more than three qubits, there are infinitely many SLOCC classes.
It is possible to partition these into families which share similar properties.
Yet, so far, there is no consensus as to how to partition the classes: there are different schemes for partitioning even the four-qubit entanglement classes, yielding different families \cite{verstraete_four_2002,lamata_inductive_2007,backens_inductive_2016}.

It is sometimes useful to generalise the definitions of GHZ and W states to $n$-qubit states.
\begin{dfn}
 The generalised GHZ state on $n$ qubits is:
 \begin{equation}
  \ket{GHZ_n} := \ket{0}\t{n} + \ket{1}\t{n}.
 \end{equation}
 The generalised W state on $n$ qubits is:
 \begin{equation}
  \ket{W_n} := \begin{cases} \ket{1} & \text{if } n= 1 \\ \ket{0}\t{n-1}\otimes\ket{1} + \ket{0}\otimes\ket{W_{n-1}} & \text{if } n>1. \end{cases}
 \end{equation}
\end{dfn}

This makes the $n$-qubit generalised GHZ state equal to $\ket{00\ldots 0} + \ket{11\ldots 1}$, i.e.\ it is the state corresponding to the $n$-ary equality signature.
The $n$-qubit generalised W state is:
\begin{equation}
 \ket{W_n} = \ket{00\ldots 01} + \ket{010\ldots 0} + \ldots + \ket{100\ldots 0},
\end{equation}
i.e.\ it corresponds to the indicator function for inputs of Hamming weight 1.

We sometimes drop the word `generalised' when talking about generalised GHZ or W states.
It should generally be clear from context whether or not we mean the three-qubit state specifically.

\subsection{Identifying types of three-qubit entanglement}
\label{s:li_et_al}

Li \emph{et al.} \cite{li_simple_2006} give an equational procedure for identifying the entanglement class of a three-qubit state from the its coefficients in the computational basis.
We recap their results here.

Let $\ket{\psi}$ be a three-qubit state and write:
\begin{equation}
 \ket{\psi} = a_0\ket{000} + a_1\ket{001} + a_2\ket{010} + a_3\ket{011} + a_4\ket{100} + a_5\ket{101} + a_6\ket{110} + a_7\ket{111}.
\end{equation}
Then:
\begin{itemize}
 \item $\ket{\psi}$ is in the GHZ class if and only if:
  \begin{equation}\label{eq:GHZ_criterion}
   (a_0a_7 - a_2a_5 + a_1a_6 - a_3a_4)^2 - 4(a_2a_4-a_0a_6)(a_3a_5-a_1a_7) \neq 0.
  \end{equation}
 \item $\ket{\psi}$ is in the W class if and only if \eqref{eq:GHZ_criterion} is not satisfied and additionally the following holds:
  \begin{equation}\label{eq:W_criterion}
   (a_0a_3\neq a_1a_2 \vee a_5a_6\neq a_4a_7) \wedge (a_1a_4\neq a_0a_5 \vee a_3a_6\neq a_2a_7) \wedge (a_3a_5\neq a_1a_7 \vee a_2a_4 \neq a_0a_6).
  \end{equation}
\end{itemize}
There are also inequalities for identifying the not-genuinely-entangled classes, which we do not use in this paper.

\subsection{The inductive entanglement classification}
\label{s:inductive_classification}

One approach for partitioning the entanglement classes of multi-qubit states into families  is the inductive entanglement classification by Lamata \emph{et al.} \cite{lamata_inductive_2006}.
In this scheme, the partition of $n$-qubit entangled states relies on the partition of $(n-1)$-qubit entangled states; hence the name.
The idea is the following.

Consider an $n$-qubit genuinely entangled state $\ket{\psi}$.
This state can be written as:
\begin{equation}
 \ket{\psi} = \ket{0}\ket{\phi_0} + \ket{1}\ket{\phi_1},
\end{equation}
where $\ket{\phi_0}, \ket{\phi_1}$ are $(n-1)$-qubit states.
As $\ket{\psi}$ is entangled, $\ket{\phi_0}$ and $\ket{\phi_1}$ must be linearly independent.
Families of entanglement classes can now be defined according to the types of entangled vectors found in $\fW = \spans\{\ket{\phi_0}, \ket{\phi_1}\}$.

In the three-qubit case, the GHZ class contains the states for which $\fW$ contains two linearly independent product vectors while the W class contains the states for which any basis of $\fW$ contains at least one entangled vector \cite{lamata_inductive_2006}.

In the four-qubit case, the inductive classification yields 10 families of genuinely entangled states \cite{lamata_inductive_2007, backens_inductive_2016}.
For each family, one can derive a set of representative states, each of which contains the smallest possible number of free parameters, such that any member of the family can be reduced to a representative state by SLOCC.

As an example, one of the four-qubit entanglement families is that of four-qubit GHZ states, which are of the form:
\begin{equation}
 (A\otimes B\otimes C\otimes D)(\ket{0000}+\ket{1111})
\end{equation}
for some complex invertible 2 by 2 matrices $A,B,C,D$.
Another family, usually labelled $\fW_{0_k\Psi,0_k\Psi}$, has representatives:
\begin{gather*}
 \ket{0000}+\ket{1100}+\lambda\ket{0011}+\mu\ket{1111} \quad\text{and} \\
 \ket{0000}+\ket{1100}+\lambda\ket{0001}+\lambda\ket{0010}+\mu\ket{1101}+\mu\ket{1110},
\end{gather*}
where $\lambda,\mu\in\CC$ with $\lambda\neq\mu$.
Any state in the family is equivalent under SLOCC to one of the representatives with appropriate values of $\lambda$ and $\mu$.

More information about the inductive entanglement classification is given in Appendix \ref{s:appendix}, where details of the classification of four- and five-qubit states are used.

\subsection{The existing results in the quantum picture}
\label{s:existing_quantum}

Several of the existing dichotomies have a very simple description in the quantum picture.

The tractable cases of the \textsc{Holant}$^*$ dichotomy (cf.\ Section \ref{s:Holant_star}) can be described as follows:
\begin{itemize}
 \item either there is no multipartite entanglement -- this corresponds to the case $\cF\subseteq\avg{\cT}$, or
 \item there is GHZ-type multipartite entanglement but it is impossible to produce W-type multipartite entanglement via gadgets -- this corresponds to the cases $\cF\subseteq\avg{O\circ\cE}$ or $\cF\subseteq\avg{K\circ\cE}$, or 
 \item there is W-type multipartite entanglement and it is impossible to produce GHZ-type multipartite entanglement via gadgets -- this corresponds to the case $\cF\subseteq\avg{K\circ\cM}$ or $\cF\subseteq\avg{KX\circ\cM}$.
\end{itemize}
By GHZ-type entanglement we mean states that are equivalent to generalised GHZ states under SLOCC, and similarly for W-type entanglement.

The tractable case of \csp{} (cf.\ Theorem \ref{thm:csp}) that does not appear in \textsc{Holant}$^*$ is also easy to describe in the quantum picture: in quantum theory, the states corresponding to affine signatures are known as \emph{stabilizer states} \cite{dehaene_clifford_2003}.
These states and the associated operations play an important role in the context of quantum error-correcting codes \cite{gottesman_heisenberg_1998} and are thus at the core of most attempts to build large-scale quantum computers \cite{devitt_quantum_2013}.
Nevertheless, the fragment of quantum theory consisting of stabilizer states and operations that preserve the set of stabilizer states is efficiently simulable on a classical computer \cite{gottesman_heisenberg_1998}; this result is known as the Gottesman-Knill theorem.

In other words, the Holant problem and quantum information theory are linked not only by quantum algorithms being an inspiration for holographic algorithms.
Instead, the known-to-be tractable signature sets of various Holant problems correspond directly to sets of states that are of independent interest in quantum computation and quantum information theory.

%% file: ternary.tex
\section{\textsc{Holant}$^+$}
\label{s:Holant_plus}

The new family of Holant problems, called \textsc{Holant}$^+$, sits in between \textsc{Holant}$^*$ and \textsc{Holant}$^c$: it has only a small number of free signatures, which are all unary.
Yet, using results from quantum theory, these can be shown to be sufficient for constructing the gadgets required to use Theorem \ref{thm:GHZ-state} and Lemma \ref{lem:W-state}.

Formally, this variant of the Holant problem is defined as follows:
\begin{equation}
 \Holp[+]{\cF} = \Holp{\cF\cup\{\ket{0},\ket{1},\ket{+},\ket{-}\}},
\end{equation}
where $\ket{+}=\ket{0}+\ket{1}$ corresponds to the `unary equality function' and $\ket{-}=\ket{0}-\ket{1}$ is a vector that is orthogonal to $\ket{+}$.
In quantum theory, the set $\{\ket{+},\ket{-}\}$ (or their normalised equivalents) are known as the \emph{Hadamard basis}, since they are related to the computational basis vectors by a Hadamard transformation: $\{\ket{+},\ket{-}\}\doteq H\circ\{\ket{0},\ket{1}\}$, where `$\doteq$' means equality up to scalar factor and:
\begin{equation}
 H = \frac{1}{\sqrt{2}}\begin{pmatrix}1&1\\1&-1\end{pmatrix}.
\end{equation}

\subsection{Why these free signatures?}

The definition of \textsc{Holant}$^+$ is motivated by the following result from quantum theory, which we give here in updated notation.

\begin{thm}[\cite{popescu_generic_1992},\cite{gachechiladze_addendum_2016}]\label{thm:popescu-rohrlich}
 Let $\ket{\Psi}$ be an $n$-system entangled state. For any two of the $n$ systems, there exists a projection, onto a tensor product of states of the other $(n-2)$ systems, that leaves the two systems in an entangled state.
\end{thm}

The original proof of this statement in \cite{popescu_generic_1992} was flawed but it has recently been corrected \cite{gachechiladze_addendum_2016}.
The following corollary is not stated explicitly in either paper, but can be seen to hold by inspecting the proof in \cite{gachechiladze_addendum_2016}.

\begin{cor}
 Let $\ket{\Psi}$ be an $n$-qubit entangled state. For any two of the $n$ qubits, there exists a projection of the other $(n-2)$ qubits onto a tensor product of computational and Hadamard basis states that leaves the two qubits in an entangled state.
\end{cor}

In other words, Theorem \ref{thm:popescu-rohrlich} holds when the systems are restricted to qubits and the projectors are restricted to be products of computational and Hadamard basis states.
Here, it is crucial to have projectors taken from two bases that are linked by the Hadamard transformation -- the corollary works only in that case.

Using the inductive entanglement classification, we now extend this result to the following theorem, which is proved in Section \ref{s:proof:three-qubit-entanglement}.

\begin{thm}\label{thm:three-qubit-entanglement}
 Let $\ket{\Psi}$ be an $n$-qubit entangled state with $n\geq 3$. There exists some choice of three of the $n$ qubits and a projection of the other $(n-3)$ qubits onto a tensor product of computational and Hadamard basis states that leaves the three qubits in a genuinely entangled state.
\end{thm}

This result is stronger than Theorem \ref{thm:popescu-rohrlich} in that we construct entangled three-qubit states rather than two-qubit ones, but on the other hand we do not require the result to hold for all choices of three qubits: all we require is the existence of some choice of three qubits for which it does hold.

\subsection{The dichotomy theorem}

Using Theorem \ref{thm:three-qubit-entanglement}, we prove our main result, a dichotomy for \textsc{Holant}$^+$ applying to complex, not necessarily symmetric signatures.
This is the only dichotomy with these properties where only a finite number of signatures are assumed to be freely available, other than the dichotomy for \#\textsc{R$_3$-CSP} \cite{cai_complexity_2014}.

\begin{thm}\label{thm:main}
 Let $\cF$ be a set of complex signatures. $\Holp[+]{\cF}$ is in \FP{} if $\cF$ satisfies one of the following conditions:
 \begin{itemize}
  \item $\Holp[*]{\cF}$ is in \FP, or
  \item $\cF\subseteq\cA$.
 \end{itemize}
 In all other cases, the problem is \sP-hard.
\end{thm}

The tractable cases are almost the same as those for symmetric \textsc{Holant}$^c$ (see Theorem \ref{thm:Holant-c}), now without the restriction to symmetric signatures.
The only difference is that the holographic transformations allowed in the affine case of the \textsc{Holant}$^c$ dichotomy are trivial in the case of \textsc{Holant}$^+$: any transformation that maps $\{\ket{=_2}, \ket{0}, \ket{1}, \ket{+}, \ket{-}\}$ to a subset of $\cA$ must itself be in $\cA$.

The tractability proof for Theorem \ref{thm:main} follows immediately by reduction to \textsc{Holant}$^*$ or \csp{}, respectively.
For the hardness proof, we use Theorem \ref{thm:three-qubit-entanglement} to construct signatures corresponding to three-qubit entangled states.
We then show that, unless we are in one of the tractable cases, it is possible to construct ternary gadgets with non-degenerate symmetric signatures.
If the ternary symmetric signature is of GHZ type, Theorem \ref{thm:GHZ-state} applies.
If the ternary symmetric signature is of W type but not in $K\circ\cM$ or $KX\circ\cM$, we use Lemma \ref{lem:W-state}.
Finally, if the ternary symmetric signature is contained in $K\circ\cM$ (or $KX\circ\cM$), then by assumption the set of available signatures $\cF$ must contain some signature that is not in $K\circ\cM$ (or $KX\circ\cM$, respectively) -- else the problem is already known to be tractable.
We show how to use such a signature to construct a binary symmetric signature that is not in $K\circ\cM$ (or $KX\circ\cM$, respectively).
Then the desired result follows by Lemma \ref{lem:W-state}.

Theorem \ref{thm:three-qubit-entanglement} is proved in Section \ref{s:proof:three-qubit-entanglement}.
The gadget constructions for ternary symmetric signatures and the associated proofs are given in Section \ref{s:symmetrising_ternary}.
The gadget construction for a symmetric binary signature that is not in $K\circ\cM$ (or $KX\circ\cM$) follows in Section \ref{s:binary}.
Section \ref{s:hardness} contains the hardness proof itself, which completes the proof of the main theorem.

\subsection{Proof of Theorem \ref{thm:three-qubit-entanglement}}
\label{s:proof:three-qubit-entanglement}

The theorem is proved inductively: we show that, given an $n$-qubit entangled state with $n>3$, it is possible to project some qubit onto a computational or Hadamard basis state in such a way that at least one of the tensor factors of the remaining state has multipartite entanglement, i.e.\ contains more than 2 qubits that are entangled with each other.

There are still several components to the inductive proof:
\begin{itemize}
 \item for $n\geq 7$, we use a generic combinatorial argument based on Theorem \ref{thm:popescu-rohrlich},
 \item for $n=6$, we use a more specialised combinatorial argument,
 \item the combinatorial arguments fail for $n<6$, so we prove the cases $n=5$ and $n=4$ using (rather lengthy) case distinctions based on the inductive entanglement classification \cite{lamata_inductive_2006,lamata_inductive_2007,backens_inductive_2016} (cf.\ Section \ref{s:inductive_classification}).
\end{itemize}
Given a non-zero state that is a tensor product of several factors, it is straightforward to construct a gadget for one of the tensor factors: each of the other tensor factors must have at least one coefficient non-zero when expressed in the computational basis.
Project the factor onto that computational basis state; the resulting scalar does not affect the complexity of a Holant problem.

As explained in Section \ref{s:quantum_states}, $\bra{\phi}$ denotes the Hermitian adjoint of $\ket{\phi}$; the inner product of two vectors $\ket{\phi}, \ket{\psi}$ in some complex Hilbert space is written as $\braket{\phi}{\psi}$. Furthermore, we write:
\begin{equation}
 \bra{\phi}_k\ket{\psi}
\end{equation}
to denote the projection of the $k$-th qubit of $\ket{\psi}$ onto the state $\ket{\phi}$.
A gadget illustrating this idea is shown in Figure \ref{fig:phi_k-psi}.
The gadget does not actually represent a complex inner product -- no Hermitian conjugates are involved -- but since computational and Hadamard basis states have real coefficients, the notation can be used anyway.

\begin{figure}
 \centering
 \input{tikz_files/phi_k-psi.tikz}
 \caption{Gadget corresponding to $\bra{\phi}_k\ket{\psi}$.}
 \label{fig:phi_k-psi}
\end{figure}

\begin{lem}\label{lem:n_geq_7}
 Let $\ket{\Psi}$ be an $n$-qubit genuinely entangled state, where $n\geq 7$. Then there exists a state $\ket{\theta}\in\left\{\ket{0},\ket{1},\ket{+},\ket{-}\right\}$ such that the state $\bra{\theta}_n\ket{\Psi}$ contains multipartite entanglement.
\end{lem}
\begin{proof}
 After projecting the $n$-th qubit onto $\ket{\theta}$, we are left with an $(n-1)$-qubit state $\ket{\Phi}=\bra{\theta}_n\ket{\Psi}$.
 As $\ket{\Psi}$ is genuinely entangled, Theorem \ref{thm:popescu-rohrlich} applies, i.e.\ for any pair of qubits there exists some projection of the remaining $(n-2)$ qubits onto a product of computational and Hadamard basis states that leaves the two original qubits in an entangled state.
 Projecting onto a product of single-qubit states does not increase entanglement.
 Thus for each size-two subset of the first $(n-1)$ qubits, there must be some choice of $\ket{\theta}$ so that the two qubits in the subset remain entangled after projection.
 Now if $n\geq 7$ then the state after projection contains at least six qubits.
 That means there are at least five different qubits the first qubit could be entangled with.
 On the other hand, there are only four different projectors.
 By the pigeonhole principle, this means there is at least one choice of $\ket{\theta}$ for which the first qubit is entangled with more than one other qubit.
 Thus the corresponding post-projection state $\ket{\Phi}$ contains multipartite entanglement.
\end{proof}

The argument from Lemma \ref{lem:n_geq_7} fails for $n=6$ as there are five qubits left after projection, which means four possible entanglement partners for any given qubit -- corresponding exactly to the four projections.
 Nevertheless, this case can also be resolved by a combinatorial argument.

\begin{lem}\label{lem:n=6}
 Let $\ket{\Psi}$ be a six-qubit genuinely entangled state. Then there exists a state $\ket{\theta}$ chosen from the set of single-qubit computational and Hadamard basis states $\left\{\ket{0}, \ket{1}, \ket{+}, \ket{-}\right\}$ such that $\bra{\theta}_6\ket{\Psi}$ contains multipartite entanglement.
\end{lem}
\begin{proof}
 Suppose, for a contradiction, that each of the four post-projection states is a tensor product of one- and two-qubit entangled states.
 This gives the following possible structures for the state $\bra{\theta}_6\ket{\Psi}$:
 \begin{itemize}
  \item a tensor product of five single-qubit states,
  \item a tensor product of one two-qubit entangled state and three single-qubit states, or
  \item a tensor product of two two-qubit entangled states and one single-qubit state.
 \end{itemize}
 There are $\binom{5}{2} = 10$ different ways of choosing two out of the first five qubits for the purposes of an argument according to Theorem \ref{thm:popescu-rohrlich}.
 If none of the states after projection contain multipartite entanglement, then even if they all have different entanglement structures containing the maximum number of two-qubit entangled factors, that only yields $4\cdot 2 = 8$ different pairs, contradicting Theorem \ref{thm:popescu-rohrlich}.
 Hence at least one of the states after projection must have multipartite entanglement.
\end{proof}

As mentioned above, these combinatorial arguments fail for $n<6$.
Instead, for $n=4$ and 5, we look at the different families of entanglement classes arising from the inductive entanglement classification, and show that in each case there is a projection that leaves the remaining qubits in an entangled state.
The proofs are long and involved; they may be found in Appendix \ref{s:appendix}.

\begin{lem}\label{lem:n=4}
 Let $\ket{\Psi}$ be a genuinely entangled four-qubit state. Then there exists a state $\ket{\theta}$ chosen from the set of single-qubit computational and Hadamard basis states $\left\{\ket{0}, \ket{1}, \ket{+}, \ket{-}\right\}$ such that $\ket{\Phi} = \bra{\theta}_1\ket{\Psi}$ is genuinely three-partite entangled.
\end{lem}

\begin{lem}\label{lem:n=5}
 Let $\ket{\Psi}$ be a genuinely entangled five-qubit state. Then there exists a state $\ket{\theta}$ chosen from the set of single-qubit computational basis states and Hadamard basis states $\left\{\ket{0}, \ket{1}, \ket{+}, \ket{-}\right\}$ such that $\ket{\Phi} = \bra{\theta}_1\ket{\Psi}$ contains multipartite entanglement.
\end{lem}

With these lemmas, we can now prove the desired result:

\begin{proof}[Proof of Theorem \ref{thm:three-qubit-entanglement}.]
 If $n=3$, no qubits are projected and the result is trivial.
 For $n=4$, 5 or 6, see Lemmas \ref{lem:n=4}, \ref{lem:n=5} and \ref{lem:n=6}, respectively.
 The case $n\geq 7$ is considered in Lemma \ref{lem:n_geq_7}.
 These lemmas can be applied repeatedly until there are only three qubits left.
\end{proof}

We have shown that we can construct a ternary non-degenerate signature in \textsc{Holant}$^+$ as long as there is some signature in the set that does contain multipartite entanglement.
If there is not, the problem is tractable anyway by the $\avg{\cT}$ case of the dichotomy for \textsc{Holant}$^*$ in Theorem \ref{thm:Holant-star}.

Yet, to be able to use the results recounted in Section \ref{s:results_ternary_symmetric}, we need to be able to construct \emph{symmetric} ternary non-degenerate signatures.

%% file: tikz_files/phi_k-psi.tikz
\begin{tikzpicture}
	\begin{pgfonlayer}{nodelayer}
		\node [style=solidn, label={below:$\ket{\psi}$}] (0) at (0, -0.75) {};
		\node [style=solidn, label={above:$\bra{\phi}$}] (1) at (0, 0.5) {};
		\node [style=none, label={above:$\scriptstyle k-1$}] (2) at (-1.5, 0.5) {};
		\node [style=none, label={above:$\scriptstyle\ldots$}] (3) at (-2.75, 0.5) {};
		\node [style=none, label={above:$\scriptstyle 1$}] (4) at (-3.5, 0.5) {};
		\node [style=none, label={above:$\scriptstyle k+1$}] (5) at (1.5, 0.5) {};
		\node [style=none, label={above:$\scriptstyle\ldots$}] (6) at (2.75, 0.5) {};
		\node [style=none, label={above:$\scriptstyle n$}] (7) at (3.5, 0.5) {};
	\end{pgfonlayer}
	\begin{pgfonlayer}{edgelayer}
		\draw [bend right=15, looseness=1.00] (4.center) to (0);
		\draw [bend right=15, looseness=1.00] (2.center) to (0);
		\draw (1) to (0);
		\draw [bend right=15, looseness=1.00] (0) to (5.center);
		\draw [bend right=15, looseness=1.00] (0) to (7.center);
	\end{pgfonlayer}
\end{tikzpicture}

%% file: symmetrising.tex
\subsection{Symmetrising ternary signatures}
\label{s:symmetrising_ternary}

To be able to use the results recounted in Section \ref{s:results_ternary_symmetric}, we require \emph{symmetric} ternary entangled signatures.
The signatures constructed according to the process outlined in the previous section are ternary and entangled, but they are not generally symmetric.
Nevertheless, as we show in this section, it is possible to use the general ternary entangled signatures to construct symmetric ones (possibly with the help of an additional binary non-degenerate signature).

We prove this by distinguishing cases according to whether the ternary signature constructed using Theorem \ref{thm:three-qubit-entanglement} is in the GHZ or W entanglement class.

First, consider a general GHZ class state $\ket{\psi}$.
This state is related to $\ket{GHZ}$ by SLOCC, i.e.\ there exist invertible complex 2 by 2 matrices $A,B,C$ such that:
\begin{equation}
 \ket{\psi} = (A\otimes B\otimes C)\ket{GHZ}.
\end{equation}
We can then draw the signature associated with $\ket{\psi}$ as a `virtual gadget':
\begin{center}
 \input{tikz_files/GHZ_class_state.tikz}
\end{center}
The `boxes' denoting the matrices are non-symmetric because the matrices will not in general be symmetric.
The white dot represents the GHZ state.
This notation is not meant to imply that the signatures $A,B,C$ or the ternary equality signature are available on their own.
It will simply make things easier to think of the signature as such a composite rather than a single object.

Three copies of $\ket{\psi}$ can be connected up to form the rotationally symmetric gadget shown in Figure \ref{fig:symmetrising_GHZ}.
In fact, the signature for that gadget is fully symmetric, in the sense that its value depends only on the Hamming weight of the inputs.
On the other hand, it may not be entangled or it may have the all-zero signature.

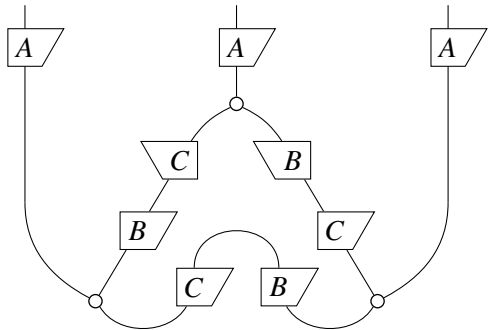
\begin{figure}
 \centering
 \input{tikz_files/symmetrising_GHZ.tikz}
 \caption{A symmetric gadget constructed from three copies of a ternary signature.}
 \label{fig:symmetrising_GHZ}
\end{figure}

For a general non-symmetric $\ket{\psi}$ there are three different such symmetric gadgets that can be constructed by permuting the roles of $A$, $B$, and $C$ in Figure \ref{fig:symmetrising_GHZ} -- in particular, which of the three ends up on the external edge of the gadget.

\begin{lem}\label{lem:GHZ_symmetrise}
 Let $\ket{\psi}$ be a three-qubit GHZ class state, i.e.\ $\ket{\psi}=(A\otimes B\otimes C)\ket{GHZ}$ for some invertible 2 by 2 matrices $A,B,C$. Then at least one of the three possible symmetric gadgets resulting from permutations of $A,B,C$ in Figure \ref{fig:symmetrising_GHZ} is non-degenerate unless $\ket{\psi}\in K\circ\cE$ and is furthermore already symmetric.
\end{lem}
\begin{proof}
 Consider the gadget in Figure \ref{fig:symmetrising_GHZ} and write:
 \begin{equation}
  M = C^T B = \begin{pmatrix}a&b\\c&d\end{pmatrix}.
 \end{equation}
 Then the signature of that gadget is:
 \begin{equation}\label{eq:symmetrising_GHZ}
  A\circ[a^3, abc, bcd, d^3].
 \end{equation}
 The SLOCC transformation by $A\t{3}$ does not change the entanglement class of the state, hence we ignore it from now on. 
 
 By the result of Li \emph{et al.} given in Section \ref{s:li_et_al}, this is a GHZ class state if and only if:
 \begin{equation}
  (ad+3bc)(ad-bc)^3 a^2 d^2\neq 0.
 \end{equation}
 Now, as $M$ is invertible, we have $ad-bc\neq 0$.
 Continuing with the analysis from Li \emph{et al.}, we find that \eqref{eq:symmetrising_GHZ} is a W SLOCC class state if either $ad+3bc=0$ or $a=0$ or $d=0$, and vanishes if $a=d=0$; other combinations of zero coefficients are excluded by the assumption of invertibility of $M$.
 
 As noted above, there are actually three different gadgets that can be constructed from the same non-symmetric three-qubit state, by cyclically permuting the roles of $A$, $B$, and $C$ in Figure \ref{fig:symmetrising_GHZ}.
 The only case in which all three gadget constructions fail is if the top left and bottom right components are zero for all of $B^TA, C^TB$,  and $A^TC$.
 In that case, there must exist invertible diagonal matrices $D_1$, $D_2$, and $D_3$ such that:
 \begin{equation}\label{eq:GHZ_fail_conditions}
  B^TA=XD_1, \qquad C^TB=XD_2, \quad\text{and}\quad A^TC = XD_3.
 \end{equation}
 This implies that $C = A X D_1^{-1} X D_2$ and $B = A X D_3^{-1} X D_2$.
 Both $X D_1^{-1} X D_2$ and $X D_3^{-1} X D_2$ are diagonal themselves; write these matrices as:
 \begin{equation}
  D_C = X D_1^{-1} X D_2 = \begin{pmatrix}\gamma_0&0\\0&\gamma_1\end{pmatrix} \quad\text{and}\quad D_B = X D_3^{-1} X D_2 = \begin{pmatrix}\beta_0&0\\0&\beta_1\end{pmatrix},
 \end{equation}
 respectively, for some $\beta_0,\beta_1,\gamma_0,\gamma_1\in\CC\setminus\{0\}$.
 Then:
 \begin{equation}
  \ket{\psi} = (A\otimes B\otimes C)\ket{GHZ} = A\t{3} (I\otimes D_B\otimes D_C) \ket{GHZ} = A\t{3} \left( \beta_0\gamma_0\ket{000} + \beta_1\gamma_1\ket{111} \right),
 \end{equation}
 where $I$ is the 2 by 2 identity matrix.
 This state is always symmetric.
 
 Furthermore, we find that:
 \begin{equation}
  A^TA = X D_3 X D_2^{-1} X D_1 = \begin{pmatrix}0&\alpha_0\\\alpha_1&0\end{pmatrix},
 \end{equation}
 where $\alpha_0,\alpha_1\in\CC\setminus\{0\}$.
 As the left-hand side of this equality is invariant under transpose, we must in fact have $\alpha_0=\alpha_1$, i.e.:
 \begin{equation}\label{eq:ATA=X}
  A^TA=\alpha_0 \begin{pmatrix}0&1\\1&0\end{pmatrix}.
 \end{equation}
 It is straightforward to check that all matrices satisfying \eqref{eq:ATA=X} must be of the form (see Appendix \ref{a:ATA=X}):
 \begin{equation}
  A = \begin{pmatrix}1&1\\i&-i\end{pmatrix} \begin{pmatrix}a_0&0\\0&a_1\end{pmatrix} \quad\text{or}\quad A = \begin{pmatrix}1&1\\i&-i\end{pmatrix} X  \begin{pmatrix}a_0&0\\0&a_1\end{pmatrix}
 \end{equation}
 for some $a_0,a_1\in\CC\setminus\{0\}$.
 But then:
 \begin{equation}
  \ket{\psi} = (A\otimes B\otimes C)\ket{GHZ} = K\t{3} \left( a_0^3\beta_0\gamma_0\ket{000} + a_1^3\beta_1\gamma_1\ket{111} \right) \in K\circ\cE.
 \end{equation}
 This completes the proof.
\end{proof}

We have shown that given a GHZ class signature, we can always construct a non-degenerate symmetric ternary signature: whenever the original signature is not already symmetric, we can use the gadget in Figure \ref{fig:symmetrising_GHZ}.
The new signature may be in the GHZ or the W class.

A similar approach works for W class signatures, though in some cases an additional binary signature is required.

\begin{lem}\label{lem:W_symmetrise}
 Let $\ket{\psi}$ be a three-qubit W class state, i.e.\ $\ket{\psi}=(A\otimes B\otimes C)\ket{W}$ for some invertible 2 by 2 matrices $A,B,C$. If $\ket{\psi}\in K\circ\cM$ (or $\ket{\psi}\in KX\circ\cM$), assume that we also have a two-qubit entangled state $\ket{\phi}$ that is not in $K\circ\cM$ (or $KX\circ\cM$, respectively). Then we can construct a symmetric three-qubit entangled state.
\end{lem}
\begin{proof}
 First note that if $\ket{\psi}\in K\circ\cM$ and $\ket{\phi}\notin K\circ\cM$, then:
 \begin{align}
  \ket{\psi} &= K\t{3} \left( a_0\ket{000} + a_1\ket{001} + a_2\ket{010} + a_4\ket{100} \right) \quad\text{and}\\
  \ket{\phi} &= K\t{2} \left( b_0\ket{00} + b_1\ket{01} + b_2\ket{10} + b_3\ket{11} \right),
 \end{align}
 where $a_1a_2a_4\neq 0$ by entanglement of $\ket{\psi}$, $b_0b_3-b_1b_2\neq 0$ by entanglement of $\ket{\phi}$ and $b_3\neq 0$ because $\ket{\phi}\notin K\circ\cM$.
 Then we can construct a gadget by connecting the second input of a vertex carrying the signature corresponding to $\ket{\phi}$ to the first input of a vertex carrying the signature corresponding to $\ket{\psi}$:
 \begin{equation}
  \input{tikz_files/not_KcM.tikz}
 \end{equation}
 This gadget has signature:
 \begin{align}
  K\t{3} \left( (a_0b_1+a_4b_0)\ket{000} + a_1b_1\ket{001} + a_1b_3\ket{010} + (a_0b_3+a_4b_2)\ket{100} + a_1b_3\ket{101} + a_2b_3\ket{110} \right)
 \end{align}
 As $a_1,a_2$ and $b_3$ must be non-zero, the coefficients of $\ket{101}$ and $\ket{110}$ are non-zero: hence this state is not in $K\circ\cM$.
 A similar argument works for $KX\circ\cM$.
 Hence we can assume that $\ket{\psi}\notin K\circ\cM \cup KX\circ\cM$ by replacing it with the above gadget if necessary.

 If the circular dots in Figure \ref{fig:symmetrising_GHZ} are W states then that diagram represents a gadget consisting of three copies of $\ket{\psi}$, again each decomposed into a `virtual gadget'.
 Let:
 \begin{equation}
  M = C^T B = \begin{pmatrix}a&b\\c&d\end{pmatrix}.
 \end{equation}
 Then the signature of the gadget in Figure \ref{fig:symmetrising_GHZ} is:
 \begin{equation}\label{eq:symmetrising_W}
  A\circ[b^3+c^3+3abd+3acd, \; ab^2+abc+ac^2+a^2d, \; a^2b+a^2c, \; a^3].
 \end{equation}
 As before, the SLOCC transformation by $A\t{3}$ does not change the entanglement class of the state, so we ignore it from now on. 
 
 By the result of Li \emph{et al.} \cite{li_simple_2006} given in Section \ref{s:li_et_al}, the gadget represents a GHZ class state if $4(ad-bc)^3 a^6\neq 0$.
 The matrix $M$ is invertible, so $\det M = ad-bc\neq 0$.
 Therefore, the gadget fails to be a GHZ state only if $a=0$.
 In that case, \eqref{eq:symmetrising_W} reduces to $[b^3+c^3,0,0,0]$, which is degenerate.
  
 For simplicity, relabel $A,B,C$ to $A_1,A_2,A_3$.
 Each of those matrices can be written in PLDU form as a product of either the identity or X, a lower triangular matrix with ones on the diagonal, an invertible diagonal matrix, and an upper triangular matrix with ones on the diagonal:
 \begin{equation}
  A_k = P_k \begin{pmatrix} 1 & 0 \\ \alpha_k & 1 \end{pmatrix} \begin{pmatrix} \beta_k & 0 \\ 0 & \gamma_k \end{pmatrix} \begin{pmatrix} 1 & \delta_k \\ 0 & 1 \end{pmatrix} = P_k \begin{pmatrix} \beta_k & \beta_k\delta_k \\ \alpha_k\beta_k & \alpha_k\beta_k\delta_k+\gamma_k \end{pmatrix}
 \end{equation}
 for some $P_k\in\{I,X\}$, $\alpha_k,\beta_k,\gamma_k,\delta_k\in\CC$ with $\beta_k, \gamma_k\neq 0$ for $k=1,2,3$.
 
 As in Lemma \ref{lem:GHZ_symmetrise}, there are different ways of combining the three non-symmetric signatures into a symmetric gadget.
 To take these into account, we assume that the matrix $M$ (the edge function between two vertices with signature $\ket{W}$) is equal to $A_i^TA_j$ for some $i,j\in\{1,2,3\}$ with $i\neq j$.
 Thus:
 \begin{equation}
  M = \begin{pmatrix} \beta_i & \alpha_i\beta_i \\ \beta_i\delta_i & \alpha_i\beta_i\delta_i+\gamma_i \end{pmatrix} P_i P_j \begin{pmatrix} \beta_j & \beta_j\delta_j \\ \alpha_j\beta_j & \alpha_j\beta_j\delta_j+\gamma_j \end{pmatrix}
 \end{equation}
 which means that:
 \begin{equation}\label{}
  M_{00} = \begin{cases} \beta_i\beta_j(1+\alpha_i\alpha_j) & \text{if } P_i=P_j \\ \beta_i\beta_j(\alpha_i+\alpha_j) & \text{if } P_i\neq P_j. \end{cases}
 \end{equation}
 Up to transpose (which does not change the value of $M_{00}$), there are three different ways of combining three copies of $\ket{\Psi}$ in a symmetrical way.
 The symmetrisation procedure fails only if none of those work, i.e.\ if:
 \begin{itemize}
  \item $P_1=P_2=P_3$ and $0=1+\alpha_1\alpha_2=1+\alpha_2\alpha_3=1+\alpha_1\alpha_3$, which means $\alpha_1=\alpha_2=\alpha_3=\pm i$, or
  \item $P_1=P_2\neq P_3$ and $0=1+\alpha_1\alpha_2=\alpha_1+\alpha_3=\alpha_2+\alpha_3$, which means $\alpha_1=\alpha_2=-\alpha_3=\pm i$, or
  \item any permutation of $1,2,3$ in the above.
 \end{itemize}
 Note that $K$ and $KX$ have the following PLDUs:
 \begin{equation}
  \begin{pmatrix} 1 & 1 \\ \pm i & \mp i \end{pmatrix} = \begin{pmatrix} 1 & 0 \\ \pm i & 1 \end{pmatrix} \begin{pmatrix} 1 & 0 \\ 0 & \mp 2i \end{pmatrix} \begin{pmatrix} 1 & 1 \\ 0 & 1 \end{pmatrix} = X \begin{pmatrix} 1 & 0 \\ \mp i & 1 \end{pmatrix} \begin{pmatrix} \pm i & 0 \\ 0 & 2 \end{pmatrix} \begin{pmatrix} 1 & -1 \\ 0 & 1 \end{pmatrix}.
 \end{equation}
 
 Consider the case $P_1=P_2=P_3=I$ and $\alpha_1=\alpha_2=\alpha_3=i$; the other cases can be treated analogously.
 We can write $A_k$ as:
 \begin{align}
  A_k &= \begin{pmatrix} 1 & 0 \\ i & 1 \end{pmatrix} \begin{pmatrix} \beta_k & 0 \\ 0 & \gamma_k \end{pmatrix} \begin{pmatrix} 1 & \delta_k \\ 0 & 1 \end{pmatrix}
  = K \begin{pmatrix} 1 & -1 \\ 0 & 1 \end{pmatrix} \begin{pmatrix} 1 & 0 \\ 0 & i/2 \end{pmatrix} \begin{pmatrix} \beta_k & 0 \\ 0 & \gamma_k \end{pmatrix} \begin{pmatrix} 1 & \delta_k \\ 0 & 1 \end{pmatrix} \\
  &= K \begin{pmatrix} \beta_k & \beta_k\delta_k - \frac{i\gamma_k}{2} \\ 0 & \frac{i\gamma_k}{2} \end{pmatrix}.
 \end{align}
 Analogous arguments work for the other failure cases: we always find that all $A_k$ correspond to a product of $K$ and some upper triangular matrix, or all correspond to a product of $KX$ with some upper triangular matrix.
 
 Now, if $U_1,U_2,U_3$ are all upper triangular, then:
 \begin{equation}
  (U_1\otimes U_2\otimes U_3)\ket{W} \in \cM.
 \end{equation}
 Thus, the symmetrisation procedure fails only if $\ket{\psi}\in K\circ\cM$ or $\ket{\psi}\in KX\circ\cM$.
 But under the conditions of the lemma, we were able to assume that $\ket{\psi}$ was not in either of those sets.
 Hence the gadget works in all required cases.
\end{proof}

%% file: tikz_files/GHZ_class_state.tikz
\begin{tikzpicture}
	\begin{pgfonlayer}{nodelayer}
		\node [style=hollown] (0) at (0, -1) {};
		\node [style=none] (1) at (-2, 1.25) {};
		\node [style=map] (2) at (-2, 0.25) {$A$};
		\node [style=map] (3) at (0, 0.25) {$B$};
		\node [style=map] (4) at (2, 0.25) {$C$};
		\node [style=none] (5) at (0, 1.25) {};
		\node [style=none] (6) at (2, 1.25) {};
	\end{pgfonlayer}
	\begin{pgfonlayer}{edgelayer}
		\draw (0) to (3);
		\draw (1.center) to (2);
		\draw [bend right, looseness=1.00] (2) to (0);
		\draw [bend right, looseness=1.00] (0) to (4);
		\draw (4) to (6.center);
		\draw (3) to (5.center);
	\end{pgfonlayer}
\end{tikzpicture}

%% file: tikz_files/symmetrising_GHZ.tikz
\begin{tikzpicture}[scale=1.5]
	\begin{pgfonlayer}{nodelayer}
		\node [style=hollown] (0) at (0, 1.25) {};
		\node [style=hollown] (1) at (-2.5, -2.25) {};
		\node [style=hollown] (2) at (2.5, -2.25) {};
		\node [style=none] (3) at (0, 3) {};
		\node [style=none] (4) at (-3.75, -0.5) {};
		\node [style=none] (5) at (3.75, -0.5) {};
		\node [style=none] (6) at (-3.75, 3) {};
		\node [style=none] (7) at (3.75, 3) {};
		\node [style=map] (8) at (-3.75, 2.25) {$A$};
		\node [style=map] (9) at (-1.75, -1) {$B$};
		\node [style=map] (10) at (-0.75, -2) {$C$};
		\node [style=map] (11) at (0.75, -2) {$B$};
		\node [style=map] (12) at (1.75, -1) {$C$};
		\node [style=map] (13) at (3.75, 2.25) {$A$};
		\node [style=map] (14) at (0, 2.25) {$A$};
		\node [style=transposemap] (15) at (-1, 0.25) {$C$};
		\node [style=transposemap] (16) at (1, 0.25) {$B$};
	\end{pgfonlayer}
	\begin{pgfonlayer}{edgelayer}
		\draw [in=-90, out=152, looseness=1.00] (1) to (4.center);
		\draw [in=-90, out=28, looseness=1.00] (2) to (5.center);
		\draw (3.center) to (0);
		\draw (6.center) to (4.center);
		\draw (7.center) to (5.center);
		\draw (1) to (9);
		\draw (2) to (12);
		\draw (9) to (15);
		\draw [bend left=15, looseness=1.00] (15) to (0);
		\draw [bend left=15, looseness=1.00] (0) to (16);
		\draw (16) to (12);
		\draw [bend right=60, looseness=1.00] (1) to (10);
		\draw [bend left=60, looseness=1.00] (2) to (11);
		\draw [bend left=90, looseness=1.50] (10) to (11);
	\end{pgfonlayer}
\end{tikzpicture}

%% file: tikz_files/not_KcM.tikz
\begin{tikzpicture}
	\begin{pgfonlayer}{nodelayer}
		\node [style=solidn] (0) at (0, -0.25) {};
		\node [style=solidn] (1) at (-0.75, 0.25) {};
		\node [style=none] (2) at (-1, 0.75) {};
		\node [style=none] (3) at (0, 0.75) {};
		\node [style=none] (4) at (1, 0.75) {};
	\end{pgfonlayer}
	\begin{pgfonlayer}{edgelayer}
		\draw [bend right=15, looseness=1.00] (2.center) to (1);
		\draw [bend right=15, looseness=1.00] (1) to (0);
		\draw (0) to (3.center);
		\draw [bend right, looseness=1.00] (0) to (4.center);
	\end{pgfonlayer}
\end{tikzpicture}

%% file: binary.tex
\subsection{Constructing binary signatures}
\label{s:binary}

We have shown in the previous section that it is possible to construct a non-degenerate ternary symmetric signature from any ternary GHZ class signature; this is Lemma \ref{lem:GHZ_symmetrise}.
Furthermore, we have shown in Lemma \ref{lem:W_symmetrise} that a similar result holds for ternary W class signatures -- though if the W class signature is in $K\circ\cM$ or $KX\circ\cM$, that result requires the presence of a non-degenerate binary signature not in that set.

Here, we show that if the full signature set $\cF$ is not a subset of $K\circ\cM$ (or $KX\circ\cM$), then it is possible to construct a symmetric binary gadget over $\cF\cup\{\ket{0}, \ket{1}, \ket{+}, \ket{-}\}$ whose signature is not in $K\circ\cM$ (or $KX\circ\cM$, respectively).
This signature can be used in Lemma \ref{lem:W_symmetrise}, and it will also be required for a hardness proof according to Theorem \ref{thm:W-state}.

\begin{lem}\label{lem:binary_notin_KcircM}
 Suppose $\ket{\psi}$ is a genuinely entangled $n$-qubit state with $n\geq 2$, and $\ket{\psi}\notin K\circ\cM$. Then there exists a non-degenerate binary gadget over $\{\ket{\psi},\ket{0},\ket{1},\ket{\pm}\}$ with signature $\ket{\varphi}\notin K\circ\cM$.
\end{lem}
\begin{proof}
 If $n=2$, $\ket{\psi}$ is the desired binary signature and we are done.
 
 Otherwise, note the following: connected gadgets over $\{\ket{\psi},\ket{0},\ket{1},\ket{\pm}\}$ containing exactly one copy of $\ket{\psi}$ are in one-to-one correspondence with connected gadgets over $\{K^{-1}\circ\ket{\psi},\ket{\pm},\ket{\pm i}\}$ containing exactly one copy of $K^{-1}\circ\ket{\psi}$.
 Here, $\ket{\pm i} = \ket{0}\pm i\ket{1}$.
 
 The problem of finding a binary gadget over $\{\ket{\psi},\ket{0},\ket{1},\ket{\pm}\}$ with signature not in $K\circ\cM$ is then equivalent to that of finding a binary gadget over $\{K^{-1}\circ\ket{\psi},\ket{\pm},\ket{\pm i}\}$ with signature not in $\cM$, as long as the gadget also satisfies the above conditions.
 
 In this case, we proceed by induction.
 We will show that, given a $k$-qubit state $\ket{\theta}$ satisfying:
 \begin{equation}\label{eq:not_cM}
  \bra{1}_1\bra{1}_2\bra{y_3}_3\ldots\bra{y_k}_k \ket{\theta} \neq 0.
 \end{equation}
 for some $y_3,\ldots,y_k\in\{0,1\}$, and:
 \begin{multline}
  \bra{0}_1\bra{0}_2\bra{\phi_3}_3\ldots\bra{\phi_k}_k \ket{\theta} \bra{1}_1\bra{1}_2\bra{\phi_3}_3\ldots\bra{\phi_k}_k \ket{\theta} \\
  - \bra{0}_1\bra{1}_2\bra{\phi_3}_3\ldots\bra{\phi_k}_k \ket{\theta} \bra{1}_1\bra{0}_2\bra{\phi_3}_3\ldots\bra{\phi_k}_k \ket{\theta} \neq 0,
  \label{eq:entangled}
 \end{multline}
 for some $\ket{\phi_3},\ldots,\ket{\phi_k}\in\CC^2$, it is possible to construct a $(k-1)$-qubit state satisfying analogous properties.
 If $k=2$, these conditions are just those of being non-degenerate and not an element of $\cM$.

 First, note that $\ket{\psi'}=K^{-1}\circ\ket{\psi}$ satisfies the properties: The assumption $\ket{\psi}\notin K\circ\cM$ implies $\ket{\psi'}\notin\cM$.
 Therefore, there exists an $n$-bit string $y$ with Hamming weight at least 2 such that $\braket{y}{\psi'}\neq 0$.
 Wlog, assume $y_1=y_2=1$ (otherwise relabel the qubits).
 Then the condition becomes equivalent to \eqref{eq:not_cM}, with $\ket{\theta}=\ket{\psi'}$.
 Furthermore, as $\ket{\psi}$ is genuinely entangled, so is $\ket{\psi'}$.
 Thus, by Theorem \ref{thm:popescu-rohrlich}, there exist $\ket{\phi_3},\ldots,\ket{\phi_n}\in\{\ket{\pm},\ket{\pm i}\}$ such that the state $\bra{\phi_3}_3\ldots\bra{\phi_n}_n \ket{\psi'}$ is entangled.
 In other words, there exist single-qubit projections such that \eqref{eq:entangled} holds for $\ket{\theta}=\ket{\psi'}$.

 Now, for the inductive step, assume we have a $k$-qubit state $\ket{\theta}$ satisfying both \eqref{eq:not_cM} and \eqref{eq:entangled}. 
 If $\ket{\phi}\in\{\ket{\pm}, \ket{\pm i}\}$, we can write $\bra{\phi}$ as $\bra{0}+\alpha\bra{1}$, where $\alpha\in\{\pm 1, \pm i\}$. Then the expression:
 \begin{multline*}
  \bra{1}_1\bra{1}_2\bra{y_3}_3\ldots\bra{y_{k-1}}_{k-1}\bra{\phi}_k \ket{\theta} \\
  = \bra{1}_1\bra{1}_2\bra{y_3}_3\ldots\bra{y_{k-1}}_{k-1}\bra{0}_k \ket{\theta} + \alpha \bra{1}_1\bra{1}_2\bra{y_3}_3\ldots\bra{y_{k-1}}_{k-1}\bra{1}_k \ket{\theta}
 \end{multline*}
 is a linear polynomial in $\alpha$. If $y_3,\ldots,y_{k-1}$ are chosen to satisfy \eqref{eq:not_cM}, the polynomial is not identically zero. Hence this expression vanishes for at most one value of $\alpha$.
 
 Furthermore, the expression:
 \begin{multline*}
  \left(\bra{0}_1\bra{0}_2\bra{\phi_3'}_3\ldots\bra{\phi_{k-1}'}_{k-1}\bra{\phi}_k \ket{\theta}\right) \left(\bra{1}_1\bra{1}_2\bra{\phi_3'}_3\ldots\bra{\phi_{k-1}'}_{k-1}\bra{\phi}_k \ket{\theta}\right) \\
  - \left(\bra{0}_1\bra{1}_2\bra{\phi_3'}_3\ldots\bra{\phi_{k-1}'}_{k-1}\bra{\phi}_k \ket{\theta}\right) \left(\bra{1}_1\bra{0}_2\bra{\phi_3'}_3\ldots\bra{\phi_{k-1}'}_{k-1}\bra{\phi}_k \ket{\theta}\right)
 \end{multline*}
 is a quadratic polynomial in $\alpha$.
 By \eqref{eq:entangled}, that polynomial is not identically zero. Hence, the expression vanishes for at most two values of $\alpha$.

 As a consequence, there must be at least one value of $\alpha$ for which both \eqref{eq:entangled} and:
 \[
  \bra{1}_1\bra{1}_2\bra{y_3}_3\ldots\bra{y_{k-1}}_{k-1}\bra{\phi}_k \ket{\theta} \neq 0
 \]
 are satisfied. For this $\alpha$, $\bra{\phi}_k\ket{\theta}$ is an $(n-1)$-qubit state which is not in $\cM$ and in which the first two qubits are entangled. If this state is not separable, repeat the procedure on the full state. If it is separable, pick the tensor factor containing the first two qubits: this will be a state not contained in $\cM$. Thus, by induction, the desired binary signature can be constructed.
\end{proof}

The binary signature required in Lemma \ref{lem:W_symmetrise} does not need to be symmetric, only non-degenerate.
Yet to show hardness using Theorem \ref{thm:W-state}, a symmetric non-degenerate binary signature is needed.

\begin{lem}\label{lem:binary-symmetric}
 Suppose $\ket{\psi}\in K\circ\cM$ is a three-qubit symmetric entangled state and $\ket{\phi}\notin K\circ\cM$ is a two-qubit entangled state. Then there exists a gadget over $\left\{ \ket{\psi}, \ket{\phi}, \ket{0}, \ket{1}, \ket{\pm} \right\}$ such that its signature $\ket{\varphi}$ is a two-qubit symmetric entangled state and $\ket{\varphi}\notin K\circ\cM$.
\end{lem}
\begin{proof}
 We can write:
 \begin{align}
  \ket{\psi} &\doteq K\t{3} \left( v\ket{000}+\ket{W_3} \right) \\
  \ket{\phi} &\doteq K\t{2} \left( \phi_{00}\ket{00}+\phi_{01}\ket{01}+\phi_{10}\ket{10}+\ket{11} \right),
 \end{align}
 where $v,\phi_{00},\phi_{01},\phi_{10}\in\CC$ and $\doteq$ denotes equality up to non-zero constant scalar factor. As $\ket{\phi}\notin K\circ\cM$, $\phi_{11}$ must be non-zero; we have normalised it to 1 for simplicity.
 
 \begin{figure}
  \centering
  \input{tikz_files/symmetric_binary.tikz}
  \caption{Gadget for a symmetric binary signature that is not in $K\circ\cM$ (or $KX\circ\cM$). The degree-1 vertex has some signature chosen from the set $\{\ket{0},\ket{1},\ket{\pm}\}$.}
  \label{fig:symmetric_binary}
 \end{figure}
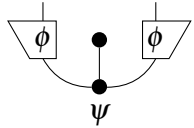

 Connecting a state from the set $\{\ket{0},\ket{1},\ket{\pm}\}$ to $\ket{\psi}$ gives a new state:
 \begin{equation}
  \ket{\psi'} = K\t{2} \left( v'\ket{00}+\ket{01}+\ket{10} \right)
 \end{equation}
 where $v'\in\{v\pm i, v\pm 1 \}$.
 Now sandwich a copy of this state between $\ket{\phi}$ and its transpose as in Figure \ref{fig:symmetric_binary}.
 That gadget has signature:
 \begin{equation}\label{eq:binary_symmetric}
  \ket{\varphi} = K\t{2} \left( \left(2\phi_{00}\phi_{01} + \phi_{01}^2 v'\right)\ket{00} + \left(\phi_{00}+\phi_{01}\phi_{10}+\phi_{01}v'\right) (\ket{01}+\ket{10}) + \left(2\phi_{10} + v'\right) \ket{11}\right).
 \end{equation}
 This signature is symmetric and non-degenerate for all $v'$.
 Furthermore, $\ket{\varphi}\in K\circ\cM$ only if $v'=-2\phi_{10}$.
 Thus, there is a choice of single-qubit state in $\{\ket{0},\ket{1},\ket{\pm}\}$ such that the gadget satisfies the requirements.
\end{proof}

An analogous argument holds with $KX$ instead of $K$.
Hence, we can construct a non-degenerate symmetric binary signature  satisfying the required properties whenever needed.

%% file: tikz_files/symmetric_binary.tikz
\begin{tikzpicture}
	\begin{pgfonlayer}{nodelayer}
		\node [label={below:$\psi$}, style=solidn] (0) at (0, -1) {};
		\node [style=solidn] (1) at (0, 0.25) {};
		\node [style=transposemap] (2) at (-1.5, 0.25) {$\phi$};
		\node [style=map] (3) at (1.5, 0.25) {$\phi$};
		\node [style=none] (4) at (-1.5, 1.25) {};
		\node [style=none] (5) at (1.5, 1.25) {};
		\node [style=none] (6) at (-1.5, -0.25) {};
		\node [style=none] (7) at (1.5, -0.25) {};
	\end{pgfonlayer}
	\begin{pgfonlayer}{edgelayer}
		\draw (4.center) to (6.center);
		\draw [bend right, looseness=1.00] (6.center) to (0);
		\draw [bend right, looseness=1.00] (0) to (7.center);
		\draw (7.center) to (5.center);
		\draw (0) to (1);
	\end{pgfonlayer}
\end{tikzpicture}

%% file: hardness.tex
\subsection{The hardness proof}
\label{s:hardness}

Suppose $\cF$ is not in one of the tractable cases.
Then, in particular, $\cF\not\subseteq\avg{\cT}$ -- i.e.\ $\cF$ must contain multipartite entanglement  (cf.\ Section \ref{s:Holant_star}).

It is therefore possible to use Theorem \ref{thm:three-qubit-entanglement} to construct a ternary entangled signature.
The quantum state associated with this signature must be either in the GHZ or in the W SLOCC class.

In the GHZ case, it is always possible to construct a non-degenerate symmetric ternary signature by Lemma \ref{lem:GHZ_symmetrise}.
In the W case, if the ternary signature is not in $K\circ\cM$ or $KX\circ\cM$, it can be used to construct a non-degenerate ternary symmetric signature on its own by Lemma \ref{lem:W_symmetrise}.
If the ternary signature is in $K\circ\cM$, by Lemma \ref{lem:binary_notin_KcircM}, we can construct a binary signature that is not in $K\circ\cM$ since $\cF\not\subseteq K\circ\cM$; and similarly with $KX$ instead of $K$.
This then enables the use of Lemma \ref{lem:W_symmetrise}.

Hence if $\cF$ is not one of the tractable sets, it is always possible to construct a non-degenerate symmetric ternary signature.
Again, the quantum state associated with this signature must be either in the GHZ or in the W SLOCC class.

If it is a GHZ class state, use the following theorem and corollary to reduce the problem to Theorem \ref{thm:GHZ-state}.
This theorem yields \sP-hardness unless $\cF$ is a subset of $\avg{O\circ\cE}$ or of $\cA$, in which cases the problem is tractable.

\begin{lem}\label{lem:partial_bipartite}
 Let $f$ be a signature and $\cG$ a set of signatures. Then:
 \begin{equation}
  \Holant(\{f\}\cup\cG) \equiv_T \Holant(\{f, [1,0,1]\} \mid \cG\cup\{[1,0,1]\}).
 \end{equation}
\end{lem}
\begin{proof}
 To see $\Holant(\{f\}\cup\cG) \leq_T \Holant(\{f, [1,0,1]\} \mid \cG\cup\{[1,0,1]\})$, consider a signature grid over $\{f\}\cup\cG$.
 This can be transformed into a bipartite signature grid over $\{f, [1,0,1]\} \mid \cG\cup\{[1,0,1]\}$ by adding a new vertex with signature $[1,0,1]$ in the middle of any edge connecting two copies of $f$ or two signatures from $\cG\setminus\{f\}$.
 If $[1,0,1]\in\cG$, the converse is trivial.
 Otherwise, to get from a bipartite signature grid over $\{f, [1,0,1]\} \mid \cG\cup\{[1,0,1]\}$ to one over $\{f\}\cup\cG$, simply remove all vertices with signature $[1,0,1]$ and merge the edges originally incident on them.

 This process can be used whether or not $f$ is in $\cG$.
 Thus the proof is complete.
\end{proof}

\begin{cor}\label{cor:partial_holographic}
 Let $f$ be a signature and $\cG$ a set of signatures, and let $M$ be an invertible 2 by 2 matrix. Then:
 \begin{equation}
  \Holant (\{M\circ f\}\cup\cG) \equiv_T \Holant\left(\left\{f, M^{-1}\circ [1,0,1]\right\} \,\middle|\, (\cG\cup\{[1,0,1]\})\circ M^T \right).
 \end{equation}
\end{cor}
\begin{proof}
 This follows from Lemma \ref{lem:partial_bipartite} and Valiant's Holant theorem.
\end{proof}

If the non-degenerate symmetric ternary signature $\ket{\psi}$ constructed in Section \ref{s:symmetrising_ternary} is in the W class, by Lemma \ref{lem:W-state}, the problem is \sP-hard unless the signature is in $K\circ\cM$ (or $KX\circ\cM$).
In the latter case, as $\cF\not\subseteq K\circ\cM$ (or $\cF\not\subseteq KX\circ\cM$), use Lemmas \ref{lem:binary_notin_KcircM} and \ref{lem:binary-symmetric} to construct a symmetric binary signature that is not in $K\circ\cM$ (or $KX\circ\cM$, respectively).

Consider the case $\ket{\psi}\in K\circ\cM$, i.e.:
\begin{equation}
 \ket{\psi} \doteq K\t{3} (v\ket{000}+\ket{W})
\end{equation}
for some $v\in\CC$.
Then the symmetric binary signature $\ket{\varphi}$ constructed in Lemma \ref{lem:binary-symmetric} is, up to scalar factor:
\begin{equation}
 K\circ \left[ 2\phi_{00}\phi_{01} + \phi_{01}^2 v', \;\;\phi_{00}+\phi_{01}\phi_{10}+\phi_{01}v', \;\; 2\phi_{10} + v' \right]
\end{equation}
for some $v', \phi_{00}, \phi_{01}, \phi_{10}\in\CC$ satisfying $\phi_{00}-\phi_{01}\phi_{10}\neq 0$ and $v'\in\{v\pm i, v\pm 1\}$.
Thus:
\begin{equation}
 \Holp{\{\ket{\varphi}\}\mid\{\ket{\psi}\}} \equiv_T \Holp{M\circ\{\ket{\varphi}\}\mid\{\ket{W}\}},
\end{equation}
where:
\begin{equation}
 M = \begin{pmatrix}1&0\\-\frac{v}{3}&1\end{pmatrix} K^T.
\end{equation}
The signature associated with $M\circ\ket{\varphi}$ is:
\begin{equation}
 \begin{pmatrix}0&1\\1&-\frac{v}{3}\end{pmatrix}\circ \left[ 2\phi_{00}\phi_{01} + \phi_{01}^2 v', \;\;\phi_{00}+\phi_{01}\phi_{10}+\phi_{01}v', \;\; 2\phi_{10} + v' \right]
\end{equation}
On input $00$, this signature takes the value $2\phi_{01}+v'$, which is non-zero by the construction of $\ket{\varphi}$.

Now, if we have a set of states $\cF$ and a symmetric ternary W-class state $\ket{\psi}$, $\Holant(\{\ket{\psi}\})$ is \sP-hard unless $\ket{\psi}\in K\circ\cM$.
But in the latter case, unless $\cF\subseteq K\circ\cM$, we can construct a binary symmetric entangled state $\ket{\varphi}\notin K\circ\cM$. Then $\Holant(\{\ket{\psi}\}\mid\{\ket{\varphi}\})$ is \sP-hard by the result of Cai \emph{et al.} \cite{cai_holant_2012} quoted in Section \ref{s:results_ternary_symmetric}.

A similar argument holds with $KX$ instead of $K$.

This concludes the \sP-hardness proof.
We have thus proved Theorem \ref{thm:main}.

%% file: conclusion.tex
\section{Conclusions}

Inspired by the connection between holographic algorithms and quantum computation, we apply knowledge from quantum information theory to Holant problems.
In particular, we reformulate existing dichotomies in the framework of quantum entanglement, leading to a concise way of expressing many known tractable classes of functions.
Motivated by this and by existing results in entanglement theory, we define a new Holant family.
This family is denoted \textsc{Holant}$^+$ and it has four freely available unary functions, including the two that are available in \textsc{Holant}$^c$.
We derive a full dichotomy for this family, which is closely related to the dichotomy for symmetric \textsc{Holant}$^c$ \cite{cai_holant_2012}.

As no full dichotomy is known for \textsc{Holant}$^c$, our dichotomy sits at the frontier of Holant research.
The similarity to the restricted \textsc{Holant}$^c$ dichotomy indicates that our result may be a useful stepping stone towards a full \textsc{Holant}$^c$ dichotomy, and thus to a full dichotomy for all Holant problems.

In deriving our \textsc{Holant}$^+$ dichotomy, we prove a new result in entanglement theory: given any $n$-qubit genuinely entangled state, it is possible to find some subset of $(n-3)$ qubits and a projector which is a tensor product of $(n-3)$ computational and Hadamard basis states so that the projection leaves the remaining three qubits in a genuinely entangled state.
This is a generalisation of a similar result about constructing two-qubit entangled states \cite{popescu_generic_1992,gachechiladze_addendum_2016}, though our result is slightly weaker: the original theorem applies for any choice of two qubits, whereas we just show that there exists some choice of three qubits for which the statement holds.
It may be possible to strengthen the argument in future work.

We expect that further analysis of Holant problems using methods from quantum information and quantum computation will lead to further new insights, both into the complexity of Holant problems and into entanglement or other areas of quantum theory.

\section*{Acknowledgements}

I would like to thank Ashley Montanaro, both for introducing me to Holant problems and for helpful comments on earlier drafts of this paper.
I acknowledge funding from EPSRC via grant EP/L021005/1.

%% file: appendix.tex
\section{Proofs of Lemmas \ref{lem:n=4} and \ref{lem:n=5}}
\label{s:appendix}

Lemmas \ref{lem:n=4} and \ref{lem:n=5} state that, given a genuinely entangled state on four or five qubits, respectively, it is always possible to project one of the qubits onto one of the states $\ket{0}, \ket{1}, \ket{+}$ or $\ket{-}$ in such a way that the remaining state still contains multipartite entanglement.

In proving these statements, we use the inductive entanglement classification by Lamata \emph{et al.}\cite{lamata_inductive_2006} (see also Section \ref{s:inductive_classification}).
This classification introduces a way of grouping the infinite number of entanglement classes into families and of determining representative states for each family.
The process is the following: an $n$-qubit state $\ket{\Psi}$ can be expressed as:
\begin{equation}
 \ket{\Psi} = \ket{0}\ket{\Phi_0} + \ket{1}\ket{\Phi_1},
\end{equation}
where $\ket{\Phi_0}$ and $\ket{\Phi_1}$ are linearly independent if $\ket{\Psi}$ is genuinely entangled.
A family is then determined by the type of $(n-1)$-qubit entanglement found in $\spans\left\{\ket{\Phi_0},\ket{\Phi_1}\right\}$.
The resulting classification is invariant under SLOCC, i.e.\ if $\ket{\Psi_1}$ and $\ket{\Psi_2}$ are related by SLOCC, then they are in the same family.

Representative states for each family are constructed by using SLOCC to eliminate as many free parameters as possible.
A complete list of entanglement families with their representative states is known for four-qubit states \cite{lamata_inductive_2007,backens_inductive_2016}.
For five-qubit states, such a list does not exist yet.
To avoid having to classify all five-qubit entangled states in detail, we use slightly different proof approaches for the two Lemmas.

\begin{proof}[Proof of Lemma \ref{lem:n=4}.]
 In this proof, we use the explicit inductive classification of four-qubit entanglement by Lamata \emph{et al.}\cite{lamata_inductive_2007}, as updated in \cite{backens_inductive_2016}.\footnote{This is also why we are now projecting the first qubit rather than the last -- as the list of cases is exhaustive, this is without loss of generality.} 
 Those two papers give 10 families of genuinely entangled four-qubit states, together with the corresponding representative states.
 
 Any four-qubit genuinely entangled state can be written as $\ket{\Psi}=(F_1\otimes F_2\otimes F_3\otimes F_4)\ket{\Theta}$, where $\ket{\Theta}$ is one of the entanglement superclass representatives and $F_1,F_2,F_3,F_4$ are invertible single-qubit operators.
 Then:
 \[
  \ket{\Phi} = \bra{\theta}_1 (F_1\otimes F_2\otimes F_3\otimes F_4) \ket{\Theta} = \left( \left(\bra{\theta} F_1\right)\otimes F_2\otimes F_3\otimes F_4 \right) \ket{\Theta}.
 \]
 Now, the operator $F_2\otimes F_3\otimes F_4$ has no effect on the entanglement classification of $\ket{\Phi}$; we will therefore ignore it from now on.
 The operator $F_1$ can be absorbed into the projector, i.e.\ let $\bra{\theta'} = \bra{\theta}F_1$.
 As $\ket{\theta} \in \left\{\ket{0}, \ket{1}, \ket{+}, \ket{-}\right\}$, the set $\left\{ \bra{\theta'} \right\}$ contains four pairwise linearly independent projectors.
 The normalisation of the projectors does not change the entanglement classification either, so choose the following:
 \[
  \bra{\theta'} = \begin{cases} \bra{0} &\text{if } \braket{\theta'}{1}=0, \\ x\bra{0} + \bra{1} &\text{otherwise.} \end{cases}
 \]
 Pairwise linear independence then implies that at most one $\bra{\theta'}$ is equal to $\bra{0}$, and the others all have distinct values for $x$.
 Sometimes we will write $\bra{\theta'}=x\bra{0}+y\bra{1}$, with the understanding that either $y=0$ or $y=1$.
 Thus it suffices to consider the representative state $\ket{\Theta}$ of each entanglement class and the modified set of projectors $\{\bra{\theta'}\}$.
 
 This leads to the following cases, taken from \cite{lamata_inductive_2007,backens_inductive_2016}.
 \begin{description}
  \item [$\mathfrak{W}_{000,000}$] The representative state is $\ket{GHZ_4}=\ket{0000}+\ket{1111}$; any non-computational basis projector yields another GHZ-type entangled state.
  
  \item [$\mathfrak{W}_{000,0_k\Psi}$] There are two representative states, $\ket{0000}+\ket{1100}+\ket{1111}$ and $\ket{0000}+\ket{1101}+\ket{1110}$. In either case, projecting the first qubit onto a non-computational basis state yields a three-qubit entangled state. (That state is in the GHZ SLOCC class for the first representative and in the W SLOCC class for the second one.)
  
  \item [$\mathfrak{W}_{000,GHZ}$] A four-qubit state is in this entanglement superclass if it can be written as:
   \[
    \left(F_1\otimes F_2\otimes F_3\otimes F_4\right) \left( \ket{0\varphi\psi\theta} + \ket{1000}+\ket{1111} \right)
   \]
   where $\spans\left\{\ket{\varphi\psi\theta},\ket{000}+\ket{111}\right\}$ contains no separable vectors other than $\ket{\varphi\psi\theta}$. Projecting the first qubit onto some state $\ket{\theta}$ satisfying $\bra{\theta}F_1=x\bra{0}+y\bra{1}$ yields:
   \[
    \left(F_2\otimes F_3\otimes F_4\right) \left(  x\ket{\varphi\psi\theta} + y\left(\ket{000}+\ket{111}\right) \right),
   \]
   i.e.\ up to SLOCC an element of $\spans\left\{\ket{\varphi\psi\theta},\ket{000}+\ket{111}\right\}$. Hence, for $y\neq 0$, this must be a three-qubit entangled state.
   
  \item [$\mathfrak{W}_{000,W}$] The representative state is $\ket{W_4} = \ket{0001}+\ket{0010}+\ket{0100}+\ket{1000}$. 
  The argument is analogous to that for $\mathfrak{W}_{000,GHZ}$.
  
  \item [$\mathfrak{W}_{0_k\Psi,0_k\Psi}$] There are two representative states (or six if one includes permutations of the last three qubits, which do not affect whether the post-projection state is genuinely entangled):
   \begin{gather*}
    \ket{0000}+\ket{1100}+\lambda\ket{0011}+\mu\ket{1111} \quad\text{or} \\
    \ket{0000}+\ket{1100}+\lambda\ket{0001}+\lambda\ket{0010}+\mu\ket{1101}+\mu\ket{1110},
   \end{gather*}
   where $\lambda\neq\mu$.
   
   In the first case, projecting the first qubit onto a state $\bra{\theta'}=x\bra{0}+y\bra{1}$ yields $x\ket{000}+\lambda x\ket{011}+y\ket{100}+\mu y\ket{111}$. This is in the GHZ SLOCC class iff $(\lambda - \mu)^2 x^2 y^2\neq 0$ (cf.\ Section \ref{s:li_et_al}), which is satisfied for any $xy\neq 0$, i.e.\ whenever $\bra{\theta'}$ is not a computational basis state.
   
   In the second case, projecting the first qubit onto a state $\bra{\theta'}$ as above yields:
   \[
    x\ket{000} + \lambda x\ket{001} + \lambda x \ket{010} + y\ket{100} + \mu y \ket{101} + \mu y \ket{110},
   \]
   which is in the W SLOCC class iff:
   \[
    \left( 0\neq \lambda^2x^2 \vee \mu^2y^2\neq 0 \right) \wedge (\lambda xy \neq \mu xy).
   \]
   The last inequality is satisfied whenever $xy\neq 0$ since $\lambda$ cannot be equal to $\mu$.
   But $\lambda,\mu$ cannot both be zero, so $xy\neq 0$ is also sufficient to satisfy at least one of the first two inequalities. Hence the state is in the W class whenever the projector is not a computational basis state.
  \item [$\mathfrak{W}_{0_i\Psi,0_j\Psi}$] A four-qubit state is in this entanglement class if it can be written as:
   \[
    \ket{\phi\varphi_1\Psi_1} + \ket{\bar{\phi}\varphi_2\psi_2\theta_2}+\ket{\bar{\phi}\bar{\varphi}_2\psi_2\bar{\theta}_2},
   \]
   where $\ket{\Psi_1}$ is a two-qubit entangled state and overbars denote linear independence. Additionally, the classification requires that $\spans\left\{\ket{\varphi_1\Psi_1}, \ket{\varphi_2\psi_2\theta_2}+\ket{\bar{\varphi}_2\psi_2\bar{\theta}_2}\right\}$ contain no fully separable states and no two bipartite separable states that are separable along the same partition. But there are only three different bipartitions, so any set of four pairwise linearly independent elements of $\spans\left\{\ket{\varphi_1\Psi_1}, \ket{\varphi_2\psi_2\theta_2}+\ket{\bar{\varphi}_2\psi_2\bar{\theta}_2}\right\}$ must contain a three-qubit entangled state.
 \end{description}
 The arguments for the four remaining cases -- $\mathfrak{W}_{0_k,GHZ}$, $\mathfrak{W}_{0_k,W}$, $\mathfrak{W}_{GHZ,W}$ and $\mathfrak{W}_{W,W}$ -- are analogous to that for $\mathfrak{W}_{000,GHZ}$.
\end{proof}

\begin{proof}[Proof of Lemma \ref{lem:n=5}.]
 The approach here is similar to, if not quite the same as, the four-qubit case.
 The main difference is that the inductive classification of five-qubit entanglement does not exist yet, and we would like to avoid having to classify all the states ourselves.
 
 Instead, we decompose any five-qubit entangled state as $\ket{0}\ket{\Phi_0}+\ket{1}\ket{\Phi_1}$, where $\ket{\Phi_0},\ket{\Phi_1}$ are non-vanishing four-qubit states, and consider the entanglement families of $\ket{\Phi_0}$ and $\ket{\Phi_1}$.
 We use the computational basis on the first qubit rather than the eigenvectors of the coefficient matrix, as used in the inductive classification.
 As we do not care about assigning a unique entanglement family to each state but only need to ensure that each state is considered at least once, this approach simplifies the process.
 
 Furthermore, unlike in the proof of Lemma \ref{lem:n=4}, we do not transform to representative states and then use a correspondingly transformed set of projectors: instead we keep the general entangled states and the projectors are taken to be computational or Hadamard basis states.
 
 Throughout, we do not explicitly analyse different permutations of the same state, as permutations of the qubits do not change whether there is multipartite entanglement or not.
 In this sense, when we give an entanglement class of the post-projection four-qubit state, the exact class is not relevant as it may change under qubit permutations; what matters is that it is a class of genuine four-partite entanglement (and this fact does not change under qubit permutations).
 
 Following \cite{lamata_inductive_2006}, we use lower-case Greek letters to denote arbitrary single-qubit states and upper-case Greek letters to denote two-qubit entangled states.
 Overbars denote states that are linearly independent, e.g.\ $\ket{\phi}$ and $\ket{\bar{\phi}}$ are required to be linearly independent.
 We also employ the labels used by Lamata \textit{et al.}: 0000 for a product state (corresponding to a degenerate signature), $00\Psi$ for the tensor product of two single-qubit states and a two-qubit state, and so on.
 Additionally, we use the label $\Phi$ for an arbitrary four-qubit entangled state -- as we only care about the presence of multipartite entanglement, we do not need to distinguish the different families of four-qubit entanglement.
 
 Thus the following cases need to be considered. 
 \begin{description}
  \item [$0000,0000$] In this case, we can write $\ket{\Phi_0}=\ket{\phi\varphi\psi\theta}$ and $\ket{\Phi_1}=\ket{\bar{\phi}\bar{\varphi}\bar{\psi}\bar{\theta}}$. Then $\ket{\Phi_0}\pm\ket{\Phi_1}$ is in the GHZ$_4$ SLOCC class.
  
  \item [$0000,00\Psi$] In this case, $\ket{\Phi_0}=\ket{\phi\varphi\psi\theta}$ and $\ket{\Phi_1}=\ket{\bar{\phi}\bar{\varphi}\Psi}$.
   As in the $\mathfrak{W}_{000,0\Psi}$ case of the classification of four-qubit states \cite{lamata_inductive_2007}, there are two options: either $\spans\left\{\ket{\psi\theta},\ket{\Psi}\right\}$ is spanned by two product vectors, or this subspace only contains one product vector. In the former case, $\ket{\Psi}$ can be written as $a\ket{\psi\theta}+b\ket{\bar{\psi}\bar{\theta}}$ with $ab\neq 0$, in the latter as $a\ket{\psi\theta}+b\left(\ket{\psi\bar{\theta}}+\ket{\bar{\psi}\theta}\right)$ with $b\neq 0$. Projecting the first qubit onto $\ket{+}$ yields:
   \[
    \ket{\phi\varphi\psi\theta}+a\ket{\bar{\phi}\bar{\varphi}\psi\theta}+b\ket{\bar{\phi}\bar{\varphi}\bar{\psi}\bar{\theta}} \quad\text{or}\quad \ket{\phi\varphi\psi\theta}+a\ket{\bar{\phi}\bar{\varphi}\psi\theta}+b\ket{\bar{\phi}\bar{\varphi}\psi\bar{\theta}}+b\ket{\bar{\phi}\bar{\varphi}\bar{\psi}\theta}.
   \]
   The former is SLOCC equivalent to $\ket{0000}+\ket{1100}+\ket{1111}$, i.e.\ the state is in the $\mathfrak{W}_{000,0\Psi}$ class.
   For the latter, there are two cases: if $a=0$, it is SLOCC equivalent to $\ket{0000} + \ket{1101} + \ket{1110}$, which is also in the $\mathfrak{W}_{000,0\Psi}$ class. If $a\neq 0$ on the other hand, the SLOCC equivalence is to $\ket{0000} + \ket{1100} + \frac{b}{a}\ket{1101} + \frac{b}{a}\ket{1110}$, which is in the $\mathfrak{W}_{0_k\Psi,0_k\Psi}$ class (the equivalence is to the second canonical state with $\lambda_1=0$ and $\lambda_2=b/a$).
   
   \item [$0000,\Psi\Psi$]: In this case, $\ket{\Phi_0}=\ket{\phi\varphi\psi\theta}$ and $\ket{\Phi_1}=\ket{\Psi\Phi}$.
   Now we use the above argument about product vectors in subspaces twice, for $\spans\left\{\ket{\phi\varphi},\ket{\Psi}\right\}$ and for $\spans\left\{\ket{\psi\theta},\ket{\Phi}\right\}$.
   This gives four representative states for the class.
   Note, though, that two of the cases are related by a permutation of the qubits: namely the ones in which one of the subspaces contains two product vectors and the other does not.
   These two cases do not to be considered separately as the property of having genuine multipartite entanglement is invariant under qubit permutations.
   Hence the cases we need to distinguish are the following:
   \begin{itemize}
    \item $\ket{\Psi} = a\ket{\phi\varphi}+b\ket{\bar{\phi}\bar{\varphi}}$ and $\ket{\Phi} = c\ket{\psi\theta}+d\ket{\bar{\psi}\bar{\theta}}$ with $ab\neq 0$ and $cd\neq 0$: projecting the first qubit onto $\ket{\pm}$ yields:
     \[
      (1\pm ac) \ket{\phi\varphi\psi\theta} \pm ad\ket{\phi\varphi\bar{\psi}\bar{\theta}} \pm bc \ket{\bar{\phi}\bar{\varphi}\psi\theta} \pm bd \ket{\bar{\phi}\bar{\varphi}\bar{\psi}\bar{\theta}}
     \]
     If $1\pm ac=0$, the state is SLOCC equivalent to $\ket{0011}+\ket{1100}+\ket{1111}$, which is in the $\mathfrak{W}_{000,0_k\Psi}$ class (to see this, apply $I\otimes I\otimes X\otimes X$).
     If $1\pm ac\neq 0$, the state is SLOCC equivalent to:
     \[
      \ket{0000}+\frac{\pm ad}{1\pm ac}\ket{0011}+\ket{1100}+\frac{d}{c}\ket{1111},
     \]
     which is in the $\mathfrak{W}_{0_k\Psi,0_k\Psi}$ class.
    \item $\ket{\Psi} = a\ket{\phi\varphi}+b\ket{\bar{\phi}\bar{\varphi}}$ and $\ket{\Phi} = c\ket{\psi\theta}+d\left(\ket{\psi\bar{\theta}}+\ket{\bar{\psi}\theta}\right)$, where $ab\neq 0$ and $d\neq 0$: projecting the first qubit onto $\ket{\pm}$ yields:
     \[
      (1\pm ac)\ket{\phi\varphi\psi\theta} \pm ad\left(\ket{\phi\varphi\psi\bar{\theta}} + \ket{\phi\varphi\bar{\psi}\theta}\right) \pm bc \ket{\bar{\phi}\bar{\varphi}\psi\theta} \pm bd \left(\ket{\bar{\phi}\bar{\varphi}\psi\bar{\theta}} + \ket{\bar{\phi}\bar{\varphi}\bar{\psi}\theta}\right).
     \]
     If $1\pm ac=0$, this is SLOCC equivalent to:
     \[
      \ket{00}\ket{\Psi^+} + \ket{11}(c\ket{00}+d\ket{01}+d\ket{10}),
     \]
     which is in the $\mathfrak{W}_{0_k\Psi,0_k\Psi}$ class. In particular, the state $c\ket{00}+d\ket{01}+d\ket{10}$ is always linearly independent from $\ket{\Psi^+}=\ket{01}+\ket{10}$ as $ac=\mp 1$.
     If $1\pm ac\neq 0$, the state is SLOCC equivalent to:
     \[
      \ket{00}\left(\left(\frac{c}{d}\pm\frac{1}{ad}\right)\ket{00} + \ket{\Psi^+}\right) + \ket{11}\left(\frac{c}{d}\ket{00} + \ket{\Psi^+}\right).
     \]
     The two terms in parentheses are two entangled two-qubit states, and they are linearly independent. Hence this state is again in the $\mathfrak{W}_{0_k\Psi,0_k\Psi}$ class.
    \item $\ket{\Psi} = a\ket{\phi\varphi}+b\left(\ket{\phi\bar{\varphi}}+\ket{\bar{\phi}\varphi}\right)$ and $\ket{\Phi} = c\ket{\psi\theta}+d\left(\ket{\psi\bar{\theta}}+\ket{\bar{\psi}\theta}\right)$, where $b,d\neq 0$: projecting the first qubit onto $\ket{\pm}$ yields:
     \begin{multline*}
      (1\pm ac)\ket{\phi\varphi\psi\theta} \pm ad\left(\ket{\phi\varphi\psi\bar{\theta}} + \ket{\phi\varphi\bar{\psi}\theta}\right) \pm bc \left( \ket{\phi\bar{\varphi}\psi\theta} + \ket{\bar{\phi}\varphi\psi\theta}\right) \\
      \pm bd \left(\ket{\phi\bar{\varphi}\psi\bar{\theta}} + \ket{\phi\bar{\varphi}\bar{\psi}\theta} + \ket{\bar{\phi}\varphi\psi\bar{\theta}} + \ket{\bar{\phi}\varphi\bar{\psi}\theta}\right).
     \end{multline*}
     This can be rewritten as:
     \begin{multline*}
      \ket{\phi} \left( (1\pm ac)\ket{\varphi\psi\theta} \pm ad \ket{\varphi\psi\bar{\theta}} \pm ad \ket{\varphi\bar{\psi}\theta} \pm bc \ket{\bar{\varphi}\psi\theta} \pm bd \ket{\bar{\varphi}\psi\bar{\theta}} \pm bd \ket{\bar{\varphi}\bar{\psi}\theta}\right) \\ \pm b \ket{\bar{\phi}\varphi\Phi}
     \end{multline*}
     The three-qubit state in parentheses on the first line is always in the W SLOCC class, so the full four-qubit state is in the $\mathfrak{W}_{0\Psi,W}$ class.
   \end{itemize}

  \item [${0000,0GHZ}$] This case contains three-partite entanglement when projected onto the correct computational basis state.
 
  \item [${0000,0W}$] This case contains three-partite entanglement when projected onto the correct computational basis state.
 
  \item [${0000,\Phi}$] This case contains multipartite entanglement when projected onto the correct computational basis state.
 
  \item [${00\Psi,00\Psi}$] In this case, $\ket{\Phi_0} = \ket{\phi\varphi\Psi}$ and $\ket{\Phi_1} = \ket{\bar{\phi}\bar{\varphi}\bar{\Psi}}$. Projecting the first qubit onto $\ket{+}$ results in the state $\ket{\Phi_0}+\ket{\Phi_1}$, which is in the $\mathfrak{W}_{0_k\Psi,0_k\Psi}$ class and thus genuinely four-partite entangled.
 
  \item [${00\Psi,0\Psi0}$] Analogous to the above, projecting the first qubit onto $\ket{+}$ results in a state from the $\mathfrak{W}_{0_i\Psi,0_j\Psi}$ class, which is genuinely four-partite entangled.
  
  \item [${00\Psi,\Psi00}$] In this case, $\ket{\Phi_0} = \ket{\phi\varphi\Psi}$ and $\ket{\Phi_1} = \ket{\Phi\psi\theta}$, where $\ket{\Psi}$ and $\ket{\Phi}$ are entangled two-qubit states. As in the $\mathfrak{W}_{0000,\Psi\Psi}$ case, there are four cases, depending on whether $\spans\left\{\ket{\phi\varphi},\ket{\Phi}\right\}$ and $\spans\left\{\ket{\psi\theta},\ket{\Psi}\right\}$ contain one or two product vectors respectively. Again, the two cases where one subspace contains two product vectors and the other does not are symmetric. We therefore distinguish the following three cases.
  \begin{itemize}
   \item $\ket{\Psi}=a\ket{\psi\theta}+b\ket{\bar{\psi}\bar{\theta}}$ and $\ket{\Phi}=c\ket{\phi\varphi}+d\ket{\bar{\phi}\bar{\varphi}}$, where $ab\neq 0$ and $cd\neq 0$: projecting the first qubit onto $\ket{\pm}$ yields:
    \[
     (a\pm c)\ket{\phi\varphi\psi\theta} + b\ket{\phi\varphi\bar{\psi}\bar{\theta}}  \pm d\ket{\bar{\phi}\bar{\varphi}\psi\theta}
    \]
    As $a,c\neq 0$, for at least one choice of sign $a\pm c\neq 0$. For this choice, the state is in the $\mathfrak{W}_{000,0_k\Psi}$ class.
   \item $\ket{\Psi}=a\ket{\psi\theta}+b\left(\ket{\psi\bar{\theta}}+\ket{\bar{\psi}\theta}\right)$ and $\ket{\Phi}=c\ket{\phi\varphi}+d\ket{\bar{\phi}\bar{\varphi}}$, where $b\neq 0$ and $cd\neq 0$: projecting the first qubit onto $\ket{\pm}$ yields:
    \[
     (a\pm c)\ket{\phi\varphi\psi\theta} + b\ket{\phi\varphi\psi\bar{\theta}} + b\ket{\phi\varphi\bar{\psi}\theta} \pm d\ket{\bar{\phi}\bar{\varphi}\psi\theta}.
    \]
    This states is in the $\mathfrak{W}_{000,0_k\Psi}$ class for any values of $a,c$.
    \item $\ket{\Psi}=a\ket{\psi\theta} + b\left(\ket{\psi\bar{\theta}}+\ket{\bar{\psi}\theta}\right)$ and $\ket{\Phi}=c\ket{\phi\varphi} + d\left(\ket{\phi\bar{\varphi}}+\ket{\bar{\phi}\varphi}\right)$, where $b,d\neq 0$: projecting the first qubit onto $\ket{\pm}$ yields:
     \[
      (a\pm c)\ket{\phi\varphi\psi\theta} + b\ket{\phi\varphi\psi\bar{\theta}} + b\ket{\phi\varphi\bar{\psi}\theta} \pm d\ket{\phi\bar{\varphi}\psi\theta} \pm d\ket{\bar{\phi}\varphi\psi\theta}.
     \]
     This state is in the $\mathfrak{W}_{000,W}$ class for any values of $a,c$.
  \end{itemize}
  
  \item [${00\Psi,\Psi\Psi}$] In this case, $\ket{\Phi_0} = \ket{\phi\varphi\Psi}$ and $\ket{\Phi_1} = \ket{\Phi\bar{\Psi}}$. Again, there are four cases, depending on whether $\spans\left\{\ket{\phi\varphi},\ket{\Phi}\right\}$ and $\spans\left\{\ket{\Psi},\ket{\bar{\Psi}}\right\}$ contain one or two product vectors respectively.
  \begin{itemize}
   \item $\ket{\Phi}=r\ket{\phi\varphi}+s\ket{\bar{\phi}\bar{\varphi}}$, $\ket{\Psi} = a\ket{\psi\theta}+b\ket{\bar{\psi}\bar{\theta}}$ and $\ket{\bar{\Psi}} = c\ket{\psi\theta}+d\ket{\bar{\psi}\bar{\theta}}$, where $rs\neq 0$, $ab\neq0$, $cd\neq 0$ and $ad-bc\neq 0$:\footnote{The last condition is linear independence of $\ket{\Psi}$ and $\ket{\bar{\Psi}}$.} projecting the first qubit onto $\ket{\pm}$ yields:
    \[
     (a\pm cr)\ket{\phi\varphi\psi\theta} + (b\pm dr) \ket{\phi\varphi\bar{\psi}\bar{\theta}} \pm cs \ket{\bar{\phi}\bar{\varphi}\psi\theta} \pm ds \ket{\bar{\phi}\bar{\varphi}\bar{\psi}\bar{\theta}}
    \]
    Since $a,b,c,d,r$ are all non-zero, the two expressions $a\pm cr$ and $b\pm dr$ cannot both be zero for both choices of sign.
    When they are not both zero, the following cases arise:
    \begin{itemize}
     \item $a\pm cr=0\neq b\pm dr$: map states with an overbar to $\ket{0}$, non-barred ones to $\ket{1}$ to get:
      \[
       (b\pm dr) \ket{1100} \pm cs \ket{0011} \pm ds \ket{0000}.
      \]
      This state is in the $\mathfrak{W}_{0_k\Psi,0_k\Psi}$ class.
     \item $a\pm cr\neq 0 = b\pm dr$: map non-barred states to $\ket{0}$ and ones with an overbar to $\ket{1}$ to get:
      \[
       (a\pm cr)\ket{0000} \pm cs \ket{1100} \pm ds \ket{1111}.
      \]
      This state is in the $\mathfrak{W}_{0_k\Psi,0_k\Psi}$ class.
     \item $a\pm cr,b\pm dr \neq 0$: map non-barred states to $\ket{0}$ and ones with an overbar to $\ket{1}$ to get:
      \[
       (a\pm cr)\ket{0000} + (b\pm dr) \ket{0011} \pm cs \ket{1100} \pm ds \ket{1111}.
      \]
      The state can be further transformed to:
      \[
       \ket{0000} + \frac{b\pm dr}{a\pm cr} \ket{0011} + \ket{1100} + \frac{d}{c} \ket{1111}
      \]
      by a SLOCC operation on the first qubit that re-scales the computational basis states.
      Now it is straightforward to check that $\frac{b\pm dr}{a\pm cr}\neq\frac{d}{c}$, which implies that this state is in the $\mathfrak{W}_{0_k\Psi,0_k\Psi}$ class.
    \end{itemize}
   \item $\ket{\Phi}=r\ket{\phi\varphi}+s\ket{\bar{\phi}\bar{\varphi}}$, $\ket{\Psi} = a\ket{\psi\theta} + b\left(\ket{\psi\bar{\theta}}+\ket{\bar{\psi}\theta}\right)$ and (as $\ket{\Psi}$ and $\ket{\bar{\Psi}}$ must always take the same general form) $\ket{\bar{\Psi}} = c\ket{\psi\theta} + d\left(\ket{\psi\bar{\theta}}+\ket{\bar{\psi}\theta}\right)$, where $rs\neq 0$, $b\neq0$, $d\neq 0$ and $ad-bc\neq 0$: projecting the first qubit onto $\ket{\pm}$ yields:
    \[
     (a\pm cr)\ket{\phi\varphi\psi\theta} + (b\pm dr) \left( \ket{\phi\varphi\psi\bar{\theta}} + \ket{\phi\varphi\bar{\psi}\theta} \right) \pm cs \ket{\bar{\phi}\bar{\varphi}\psi\theta} \pm ds \left( \ket{\bar{\phi}\bar{\varphi}\psi\bar{\theta}} + \ket{\bar{\phi}\bar{\varphi}\bar{\psi}\theta} \right)
    \]
    Note that this can be rewritten as:
    \[
     \ket{\phi\varphi}\left((a\pm cr)\ket{\psi\theta} + (b\pm dr) \left( \ket{\psi\bar{\theta}} + \ket{\bar{\psi}\theta} \right) \right) \pm s \ket{\bar{\phi}\bar{\varphi}\bar{\Psi}}
    \]
    The parameters $b,d,r$ are all non-zero, so there exists a choice of sign such that $b\pm dr\neq 0$. For this choice, the two-qubit state in parentheses is entangled. Furthermore, it is linearly independent of $\ket{\bar{\Psi}}$ as $(a\pm cr)d-(b\pm dr)c = ad-bc\neq 0$. Hence the state is in the $\mathfrak{W}_{0_k\Psi,0_k\Psi}$ class.
   \item $\ket{\Phi}=r\ket{\phi\varphi}+s\left(\ket{\phi\bar{\varphi}}+\ket{\bar{\phi}\varphi}\right)$, $\ket{\Psi} = a\ket{\psi\theta}+b\ket{\bar{\psi}\bar{\theta}}$ and $\ket{\bar{\Psi}} = c\ket{\psi\theta}+d\ket{\bar{\psi}\bar{\theta}}$, where $s\neq 0$, $ab\neq0$, $cd\neq 0$ and $ad-bc\neq 0$: projecting the first qubit onto $\ket{\pm}$ yields:
    \[
     (a\pm cr)\ket{\phi\varphi\psi\theta} + (b\pm dr) \ket{\phi\varphi\bar{\psi}\bar{\theta}} \pm cs \left(\ket{\phi\bar{\varphi}\psi\theta} + \ket{\bar{\phi}\varphi\psi\theta}\right) \pm ds \left( \ket{\phi\bar{\varphi}\bar{\psi}\bar{\theta}} + \ket{\bar{\phi}\varphi\bar{\psi}\bar{\theta}} \right)
    \]
    Continue the inductive classification according to the last (rather than first) qubit. This yields:
    \[
      \left( (a\pm cr)\ket{\phi\varphi} \pm cs \ket{\phi\bar{\varphi}} + \ket{\bar{\phi}\varphi}\right)\ket{\psi\theta} + \left( (b\pm dr) \ket{\phi\varphi} \pm ds \left( \ket{\phi\bar{\varphi}} + \ket{\bar{\phi}\varphi} \right) \right) \ket{\bar{\psi}\bar{\theta}}.
    \]
    The two terms in parentheses are entangled two-qubit states as $cds\neq 0$. They are furthermore linearly independent. Thus the state is entangled and (up to a permutation of the qubits) in the $\mathfrak{W}_{0_k\Psi,0_k\Psi}$ class.
   \item $\ket{\Phi}=r\ket{\phi\varphi}+s\left(\ket{\phi\bar{\varphi}}+\ket{\bar{\phi}\varphi}\right)$ with $s\neq 0$, $\ket{\Psi} = a\ket{\psi\theta} + b\left(\ket{\psi\bar{\theta}}+\ket{\bar{\psi}\theta}\right)$ and $\ket{\bar{\Psi}} = c\ket{\psi\theta} + d\left(\ket{\psi\bar{\theta}}+\ket{\bar{\psi}\theta}\right)$, where $b\neq0$, $d\neq 0$ and $ad-bc\neq 0$: projecting the first qubit onto $\ket{\pm}$ yields:
    \begin{multline*}
     (a\pm cr)\ket{\phi\varphi\psi\theta} + (b\pm dr) \left( \ket{\phi\varphi\psi\bar{\theta}} + \ket{\phi\varphi\bar{\psi}\theta} \right) \pm cs \left(\ket{\phi\bar{\varphi}\psi\theta} + \ket{\bar{\phi}\varphi\psi\theta}\right) \\
     \pm ds \left( \ket{\phi\bar{\varphi}\psi\bar{\theta}} + \ket{\phi\bar{\varphi}\bar{\psi}\theta} + \ket{\bar{\phi}\varphi\psi\bar{\theta}} + \ket{\bar{\phi}\varphi\bar{\psi}\theta} \right)
    \end{multline*}
    This state is in the $\mathfrak{W}_{000,W}$ class.
  \end{itemize}
  
  \item [${0\Psi0,\Psi\Psi}$] In this case, $\ket{\Phi_0} = \ket{\phi\varphi\psi\theta}+\ket{\phi\bar{\varphi}\bar{\psi}\theta}$ and $\ket{\Phi_1} = \ket{\Phi\Psi}$. Let $\ket{\bar{\phi}}$ be some state that is linearly independent of $\ket{\phi}$, and similarly pick some $\ket{\bar{\theta}}$. Then there exist complex numbers $r,s,t,u,a,b,c,d$ such that:
  \begin{align*}
   \ket{\Phi} &= r\ket{\phi\varphi}+s\ket{\phi\bar{\varphi}}+t\ket{\bar{\phi}\varphi}+u\ket{\bar{\phi}\bar{\varphi}} \\
   \ket{\Psi} &= a\ket{\psi\theta}+b\ket{\psi\bar{\theta}}+c\ket{\bar{\psi}\theta}+d\ket{\bar{\psi}\bar{\theta}}
  \end{align*}
  where $ru-st\neq 0$ and $ad-bc\neq 0$. Projecting the first qubit onto $\ket{\pm}$ yields (up to scalar factor):
  \begin{multline*}
   (ar\pm 1)\ket{\phi\varphi\psi\theta} + br\ket{\phi\varphi\psi\bar{\theta}} + cr\ket{\phi\varphi\bar{\psi}\theta} + dr\ket{\phi\varphi\bar{\psi}\bar{\theta}} + as\ket{\phi\bar{\varphi}\psi\theta} + bs\ket{\phi\bar{\varphi}\psi\bar{\theta}} \\ + (cs\pm 1)\ket{\phi\bar{\varphi}\bar{\psi}\theta} + ds\ket{\phi\bar{\varphi}\bar{\psi}\bar{\theta}} + at\ket{\bar{\phi}\varphi\psi\theta} + bt\ket{\bar{\phi}\varphi\psi\bar{\theta}} + ct\ket{\bar{\phi}\varphi\bar{\psi}\theta} + dt\ket{\bar{\phi}\varphi\bar{\psi}\bar{\theta}} \\ + au\ket{\bar{\phi}\bar{\varphi}\psi\theta} + bu\ket{\bar{\phi}\bar{\varphi}\psi\bar{\theta}} + cu\ket{\bar{\phi}\bar{\varphi}\bar{\psi}\theta} + du\ket{\bar{\phi}\bar{\varphi}\bar{\psi}\bar{\theta}}
  \end{multline*}
  Proceeding with the inductive classification, this can be thought of as $\ket{\phi}\ket{\Phi_0'}+\ket{\bar{\phi}}\ket{\Phi_1'}$. Then $\ket{\Phi_1'}$ is always bipartite entangled of type $0_1\Psi$. The state $\ket{\Phi_0'}$ is in the GHZ SLOCC class if $br+ds\neq 0$. If $br+ds=0$, it is in the W SLOCC class as long as all of the following statements hold:
  \[
   \begin{array}{rclcrcl}
   ((ad-bc)r\pm d)r & \neq & 0 & \vee & 0 &\neq & ((ad-bc)s\mp b)s \\
   0 & \neq & bs & \vee & dr & \neq & 0 \\
   0 & \neq & 0 & \vee & 0 & \neq & ar+cs\pm 1.
   \end{array}
  \]
  As $b,d$ cannot both be zero and $r,s$ cannot both be zero, the non-GHZ cases are the following:
  \begin{itemize}
   \item $b=0=s$ (which implies $adrt\neq 0$):
     \begin{itemize}
      \item if $ar\pm 1\neq 0$, the state is in the W SLOCC class,
      \item if $ar\pm 1= 0$, it is in the $0_2\Psi$ class,
     \end{itemize}
   \item $d=0=r$ (which implies $bcsu\neq 0$):
     \begin{itemize}
      \item if $cs\pm 1\neq 0$, the state is in the W SLOCC class,
      \item if $cs\pm 1= 0$, it is in the $0_2\Psi$ class,
     \end{itemize}
   \item $bdrs\neq 0$ and $br+ds=0$:
     \begin{itemize}
      \item if $(ad-bc)r\pm d\neq 0$, the state is in the W SLOCC class,
      \item if $(ad-bc)r\pm d= 0$, it is in the $0_2\Psi$ class.
     \end{itemize}
  \end{itemize}
  In any case, the four-qubit state is entangled.
 
 \item [${00\Psi,0GHZ}$] This case trivially contains three-partite entanglement.
 
 \item [${00\Psi,GHZ0}$] This case trivially contains three-partite entanglement.
 
 \item [${00\Psi,0W}$] This case trivially contains three-partite entanglement.
 
 \item [${00\Psi,W0}$] This case trivially contains three-partite entanglement.
 
 \item [${00\Psi,\Phi}$] This case trivially contains multipartite entanglement.
 
  \item [${\Psi_{ij}\Psi_{kl},\Psi_{ij}\Psi_{kl}}$] Wlog let $i=1$, $j=2$, $k=3$, $l=4$; the argument works analogously for any other choice of indices. In this case, $\ket{\Phi_0} = \ket{\Phi\Psi}$ and $\ket{\Phi_1} = \ket{\bar{\Phi}\bar{\Psi}}$. There are four cases depending on whether $\spans\left\{\ket{\Phi},\ket{\bar{\Phi}}\right\}$ and $\spans\left\{\ket{\Psi},\ket{\bar{\Psi}}\right\}$ contain one or two product vectors, respectively. As before, the cases where one subspace contains two product vectors and the other contains one are symmetric; so it suffices to consider three cases.
  \begin{itemize}
   \item $\ket{\Phi}=r\ket{\phi\varphi}+s\ket{\bar{\phi}\bar{\varphi}}$, $\ket{\bar{\Phi}}=t\ket{\phi\varphi}+u\ket{\bar{\phi}\bar{\varphi}}$, $\ket{\Psi}=a\ket{\psi\theta}+b\ket{\bar{\psi}\bar{\theta}}$ and $\ket{\bar{\Psi}}=c\ket{\psi\theta}+d\ket{\bar{\psi}\bar{\theta}}$, where $abcdrstu\neq 0$, $ru-st\neq 0$ and $ad-bc\neq 0$: projecting the first qubit onto $\ket{\pm}$ yields:
    \begin{multline*}
     (ar\pm ct)\ket{\phi\varphi\psi\theta} + (br\pm dt)\ket{\phi\varphi\bar{\psi}\bar{\theta}} + (as\pm cu) \ket{\bar{\phi}\bar{\varphi}\psi\theta} + (bs\pm du) \ket{\bar{\phi}\bar{\varphi}\bar{\psi}\bar{\theta}} \\
     = \ket{\phi\varphi}\left((ar\pm ct)\ket{\psi\theta} + (br\pm dt)\ket{\bar{\psi}\bar{\theta}}\right) + \ket{\bar{\phi}\bar{\varphi}}\left((as\pm cu) \ket{\psi\theta} + (bs\pm du) \ket{\bar{\psi}\bar{\theta}}\right)
    \end{multline*}
    It is straightforward to check that:
    \begin{equation}\label{eq:PsiPsi+PsiPsi_l.i.}
     (ar\pm ct)(bs\pm du)-(br\pm dt)(as\pm cu) = \pm (ad-bc)(ru-st) \neq 0,
    \end{equation}
    hence the two states in parentheses are always linearly independent. Thus, depending on the values of the parameters, the state belongs to the $\mathfrak{W}_{0_k\Psi,0_k\Psi}$ class (if none of the four terms in parentheses in the top row vanish), the $\mathfrak{W}_{000,0_k\Psi}$ class (if one of those terms vanishes), or the $\mathfrak{W}_{000,000}$ class (if two vanish). Because of the linear independence condition, the state after projection is always genuinely entangled.
   \item $\ket{\Phi}=r\ket{\phi\varphi}+s\ket{\bar{\phi}\bar{\varphi}}$, $\ket{\bar{\Phi}}=t\ket{\phi\varphi}+u\ket{\bar{\phi}\bar{\varphi}}$, $\ket{\Psi}=a\ket{\psi\theta}+b\left(\ket{\psi\bar{\theta}}+\ket{\bar{\psi}\theta}\right)$ and $\ket{\bar{\Psi}}=c\ket{\psi\theta}+d\left(\ket{\psi\bar{\theta}}+\ket{\bar{\psi}\theta}\right)$, where $bdrstu\neq 0$, $ru-st\neq 0$ and $ad-bc\neq 0$: projecting the first qubit onto $\ket{\pm}$ yields:
    \begin{multline*}
     (ar\pm ct)\ket{\phi\varphi\psi\theta} + (br\pm dt)\left(\ket{\phi\varphi\psi\bar{\theta}}+\ket{\phi\varphi\bar{\psi}\theta}\right) + (as\pm cu) \ket{\bar{\phi}\bar{\varphi}\psi\theta} \\
     + (bs\pm du) \left( \ket{\bar{\phi}\bar{\varphi}\psi\bar{\theta}} + \ket{\bar{\phi}\bar{\varphi}\bar{\psi}\theta} \right)
    \end{multline*}
    The argument then proceeds as in the previous case.
   \item $\ket{\Phi}=r\ket{\phi\varphi}+s\left(\ket{\phi\bar{\varphi}}+\ket{\bar{\phi}\varphi}\right)$, $\ket{\bar{\Phi}}=t\ket{\phi\varphi}+u\left(\ket{\phi\bar{\varphi}}+\ket{\bar{\phi}\varphi}\right)$, $\ket{\Psi}=a\ket{\psi\theta}+b\left(\ket{\psi\bar{\theta}}+\ket{\bar{\psi}\theta}\right)$ and $\ket{\bar{\Psi}}=c\ket{\psi\theta}+d\left(\ket{\psi\bar{\theta}}+\ket{\bar{\psi}\theta}\right)$, where $bdsu\neq 0$, $ru-st\neq 0$ and $ad-bc\neq 0$: projecting the first qubit onto $\ket{\pm}$ yields:
    \begin{multline*}
     (ar\pm ct)\ket{\phi\varphi\psi\theta} + (br\pm dt)\left(\ket{\phi\varphi\psi\bar{\theta}}+\ket{\phi\varphi\bar{\psi}\theta}\right) + (as\pm cu) \left(\ket{\phi\bar{\varphi}\psi\theta} + \ket{\bar{\phi}\varphi\psi\theta} \right) \\
     + (bs\pm du) \left( \ket{\phi\bar{\varphi}\psi\bar{\theta}} + \ket{\phi\bar{\varphi}\bar{\psi}\theta} +  \ket{\bar{\phi}\varphi\psi\bar{\theta}} + \ket{\bar{\phi}\varphi\bar{\psi}\theta} \right)
    \end{multline*}
    Continuing with the inductive classification, we need to classify the states:
    \begin{multline*}
     (ar\pm ct)\ket{\varphi\psi\theta} + (br\pm dt)\ket{\varphi\psi\bar{\theta}} + (br\pm dt)\ket{\varphi\bar{\psi}\theta} + (as\pm cu) \ket{\bar{\varphi}\psi\theta} \\ + (bs\pm du)\ket{\bar{\varphi}\psi\bar{\theta}} + (bs\pm du) \ket{\bar{\varphi}\bar{\psi}\theta}
    \end{multline*}
    and $\ket{\varphi}\left((as\pm cu)\ket{\psi\theta}+(bs\pm du) \left(  \ket{\psi\bar{\theta}} + \ket{\bar{\psi}\theta} \right)\right)$.
    From \eqref{eq:PsiPsi+PsiPsi_l.i.}, we deduce that $as\pm cu$ and $bs\pm du$ cannot be zero at the same time, similarly $br\pm dt$ and $bs\pm du$.
    Thus the latter state cannot vanish completely; it is either of type $000$ (if $bs\pm du=0$), or $0\Psi$. 
    The former is always a W SLOCC class state.
    Hence the four-qubit post-projection state belongs to either the $\mathfrak{W}_{000,W}$ or the $\mathfrak{W}_{0\Psi,W}$ class.
  \end{itemize}
  
  \item [${\Psi_{ij}\Psi_{kl},\Psi_{ik}\Psi_{jl}}$] Wlog let $i=1$, $j=3$, $k=2$, $l=4$; the argument works analogously for any other choice of indices. We can find single-qubit states (labelled using the usual conventions) so that:
  \[
   \ket{\Phi_0} = \ket{\phi\varphi\psi\theta} + \ket{\phi\bar{\varphi}\psi\bar{\theta}} +  \ket{\bar{\phi}\varphi\bar{\psi}\theta} + \ket{\bar{\phi}\bar{\varphi}\bar{\psi}\bar{\theta}}.
  \]
  Then $\ket{\Phi_1} = \ket{\Phi\Psi}$ with:
  \begin{align*}
   \ket{\Phi} &= r\ket{\phi\varphi}+s\ket{\phi\bar{\varphi}}+t\ket{\bar{\phi}\varphi}+u\ket{\bar{\phi}\bar{\varphi}} \\
   \ket{\Psi} &= a\ket{\psi\theta}+b\ket{\psi\bar{\theta}}+c\ket{\bar{\psi}\theta}+d\ket{\bar{\psi}\bar{\theta}},
  \end{align*}
  where $ru-st\neq 0$ and $ad-bc\neq 0$. Projecting the first qubit onto $\ket{\pm}$ yields (up to scalar factor):
  \begin{multline*}
   (ar\pm 1)\ket{\phi\varphi\psi\theta} + br\ket{\phi\varphi\psi\bar{\theta}} + cr\ket{\phi\varphi\bar{\psi}\theta} + dr\ket{\phi\varphi\bar{\psi}\bar{\theta}} + as\ket{\phi\bar{\varphi}\psi\theta} + (bs\pm 1)\ket{\phi\bar{\varphi}\psi\bar{\theta}} \\ + cs\ket{\phi\bar{\varphi}\bar{\psi}\theta} + ds\ket{\phi\bar{\varphi}\bar{\psi}\bar{\theta}} + at\ket{\bar{\phi}\varphi\psi\theta} + bt\ket{\bar{\phi}\varphi\psi\bar{\theta}} + (ct\pm 1)\ket{\bar{\phi}\varphi\bar{\psi}\theta} + dt\ket{\bar{\phi}\varphi\bar{\psi}\bar{\theta}} \\ + au\ket{\bar{\phi}\bar{\varphi}\psi\theta} + bu\ket{\bar{\phi}\bar{\varphi}\psi\bar{\theta}} + cu\ket{\bar{\phi}\bar{\varphi}\bar{\psi}\theta} + (du\pm 1)\ket{\bar{\phi}\bar{\varphi}\bar{\psi}\bar{\theta}}
  \end{multline*}
  Projecting the formerly second and now first qubit onto a state $\ket{\chi}$ satisfying $\braket{\chi}{\phi}=x$ and $\braket{\chi}{\bar{\phi}}=y$ yields:
  \begin{multline}\label{eq:post_2_proj}
   (a(rx+ty)\pm x)\ket{\varphi\psi\theta}
   + b(rx+ty)\ket{\varphi\psi\bar{\theta}}
   + (c(rx+ty)\pm y)\ket{\varphi\bar{\psi}\theta}
   + d(rx+ty)\ket{\varphi\bar{\psi}\bar{\theta}} \\
   + a(sx+uy)\ket{\bar{\varphi}\psi\theta}
   + (b(sx+uy)\pm x)\ket{\bar{\varphi}\psi\bar{\theta}}
   + c(sx+uy)\ket{\bar{\varphi}\bar{\psi}\theta}
   + (d(sx+uy)\pm y)\ket{\bar{\varphi}\bar{\psi}\bar{\theta}}.
  \end{multline}
  The vector $(x,y)$ is the complex conjugate of the components of $\ket{\chi}$ in the $\left\{\ket{\phi},\ket{\bar{\phi}}\right\}$ basis, so with $\ket{\chi}\in\left\{\ket{0},\ket{1},\ket{+},\ket{-}\right\}$, we get four pairwise linearly independent vectors.
  
  Now, \eqref{eq:post_2_proj} is a GHZ state iff:
  \[
   \left( (cr+ds)x^2 + (-ar-bs+ct+du) xy - (at+bu)y^2 \right)^2 \neq 0.
  \]
  As the polynomial is the square of a quadratic form, there are at most two linearly independent null vectors $(x,y)$ unless the quadratic form vanishes. We have a choice of four pairwise linearly-independent vectors $(x,y)$, so if the polynomial is not identically zero, there is one that gives a GHZ state.
  
  The state has W type if the above inequality is false and furthermore:
  \begin{align*}
   (- (\delta r \pm d)x + (-\delta t \pm b)y)(rx + ty) \neq 0 &\vee (-\delta s\pm c)x - (\delta u \pm a)y)(sx + uy) \neq 0 \\
   ((ar+bs\pm 1)x + (at+bu)y)x\neq 0 &\vee ((cr+ds)x + (ct+ du \pm 1)y)y\neq 0 \\
   (dx - by)(rx + ty)\neq 0 &\vee (cx - ay)(sx + uy)\neq 0,
  \end{align*}
  where $\delta = ad-bc$.
  
  The GHZ polynomial is identically zero if:
  \begin{equation}\label{eq:GHZ_id_zero}
   0=cr+ds=ct+du-ar-bs=at+bu.
  \end{equation}
  Under that assumption, we have:
  \begin{align*}
   \delta r \pm d &= adr-bcr \pm d = adr + bds \pm d = d(ar+bs\pm 1) \\
   -\delta t \pm b &= bct-adt \pm b = b(ar+bs-du)+bdu\pm b = b(ar+bs\pm 1) \\
   -\delta s \pm c &= bcs - ads \pm c = bcs + acr \pm c = c (ar+bs\pm 1) \\
   \delta u \pm a &= adu - bcu \pm a = a(ar+bs-ct)+act\pm a = a(ar+bs\pm 1),
  \end{align*}
  so the W conditions become:
  \begin{align*}
   -(dx-by)(rx+ty)(ar+bs\pm 1) \neq 0 &\vee (cx-ay)(sx+uy)(ar+bs\pm 1) \neq 0 \\
   (ar+bs\pm 1)x^2\neq 0 &\vee (ar+bs \pm 1)y^2\neq 0 \\
   (dx - by)(rx + ty)\neq 0 &\vee (cx - ay)(sx + uy)\neq 0.
  \end{align*}
  Note that $cr+ds=0$ implies that $ar+bs\neq 0$ by entanglement of $\ket{\Psi}$ and $\ket{\Phi}$. If the plus or minus sign are chosen so as to make $ar+bs\pm 1\neq 0$, then the middle condition is true whenever $x,y$ are not both zero.
  
  Now, unless $x=y=0$, $(rx+ty)$ and $(sx+uy)$ cannot be zero simultaneously by entanglement of $\ket{\Phi}$. Similarly, $(cx-ay)$ and $(dx-by)$ cannot be zero simultaneously unless $x=y=0$ by entanglement of $\ket{\Psi}$.
  Thus the only way the third statement is not satisfied is if $dx-by=0=sx+uy$ or $rx+ty=0=cx-ay$. That means that at least two of the four vectors $(x,y)$ do satisfy the third statement. But then the first statement is also satisfied, so there is always a choice of $(x,y)$ that results in a W type state.
 \end{description}
 Any other decomposition trivially contains multipartite entanglement in at least one of $\ket{\Phi_0}$ and $\ket{\Phi_1}$.
\end{proof}

\section{Complex QR decomposition and solutions of $A^TA\doteq X$}
\label{a:ATA=X}

We first extend the orthogonal QR decomposition from real to complex matrices.
Recall that any real square matrix $M$ can be written as $M=QR$ where $Q$ is an orthogonal matrix and $R$ is upper (or lower) triangular.
The equivalent result for complex matrices requires $Q$ to be unitary instead of orthogonal.
Nevertheless, many complex 2 by 2 matrices do admit a decomposition with a complex orthogonal matrix and an upper or lower triangular matrix.
Where this is not possible, we give an alternative decomposition using the matrix $K$ defined below instead of an orthogonal matrix.
Recall that :
\begin{equation}
 K = \begin{pmatrix}1&1\\i&-i\end{pmatrix} \quad\text{and}\quad X = \begin{pmatrix}0&1\\1&0\end{pmatrix}.
\end{equation}

\begin{lem}[Orthogonal QR decomposition for complex matrices]\label{lem:QR_decomposition}
 Let $M$ be an invertible 2 by 2 complex matrix, write $O$ for the set of all 2 by 2 complex orthogonal matrices, and let $K$ be as defined above.
 Then the following hold:
 \begin{itemize}
  \item There exists $Q\in O\cup\{K, KX\}$ such that $Q^{-1}M$ is upper triangular.
  \item There exists $Q\in O\cup\{K,KX\}$ such that $Q^{-1}M$ is lower triangular.
  \item If $Q^{-1}M$ is neither lower nor upper triangular for any orthogonal $Q$, then $M=KD$ or $M=KXD$, where $D$ is diagonal.
 \end{itemize}
\end{lem}
\begin{proof}
 Without loss of generality, a complex orthogonal matrix $Q$ can be written as:
 \begin{equation}
  Q = \begin{pmatrix} \alpha&\beta\\\pm\beta&\mp\alpha \end{pmatrix},
 \end{equation}
 where $\alpha,\beta\in\CC$ satisfy $\alpha^2+\beta^2=1$.
 Write $M$ as:
 \begin{equation}
  M = \begin{pmatrix} x&y\\z&w\end{pmatrix}.
 \end{equation}
 We assumed $M$ was invertible, so $\det M = xw-yz\neq 0$.
 
 Let:
 \begin{equation}
  R = QM = \begin{pmatrix} \alpha&\beta\\\pm\beta&\mp\alpha \end{pmatrix}\begin{pmatrix} x&y\\z&w\end{pmatrix} = \begin{pmatrix} \alpha x + \beta z & \alpha y + \beta w \\ \mp(\alpha z -\beta x) & \mp(\alpha w -\beta y) \end{pmatrix}
 \end{equation}
 For a lower triangular decomposition, we want $\alpha y + \beta w=0$.
 If $y=0$, the matrix is already lower triangular, i.e.\ the decomposition is trivial with $Q=I$.
 Otherwise, if $y\neq 0$, the condition for lower triangularity implies $\alpha=-\beta w/y$.
 Then:
 \begin{equation}
  R = \begin{pmatrix} \beta(z-xw/y) & 0 \\ \pm\beta(x+zw/y) & \pm\beta(y+w^2/y) \end{pmatrix}.
 \end{equation}
 Now $z-xw/y$ is non-zero because $M$ is invertible.
 But for complex $y,w$, it is possible to have $y+w^2/y=0$, or equivalently $y^2+w^2=0$.
 In that case, $M$ cannot be equal to $Q^TR$ because $R$ is not invertible whereas $M$ was assumed to be.
 Therefore the orthogonal decomposition fails.
 Yet note that $y^2+w^2=0$ is equivalent to $w=\pm iy$.
 If $w=iy$, we have:
 \begin{equation}
  (KX)^{-1} M = \frac{1}{2}\begin{pmatrix}1& i \\ 1&- i\end{pmatrix} \begin{pmatrix} x&y\\z& iy\end{pmatrix} = \frac{1}{2}\begin{pmatrix} x+iz & 0 \\ x-iz & 2y \end{pmatrix},
 \end{equation}
 which is lower triangular, and similarly with just $K$ if $w=-iy$.

 Now try an upper triangular decomposition, which requires $\alpha z -\beta x = 0$.
 If $z=0$, the decomposition is trivial with $Q=I$, so assume otherwise.
 Then this can be rewritten as $\alpha = \beta x/z$, and:
 \begin{equation}
  R = \begin{pmatrix} \beta(z+ x^2/z) & \beta(w + xy/z) \\ 0 & \pm\beta(y-xw/z) \end{pmatrix}.
 \end{equation}
 As before, $y-xw/z$ is always non-zero, but $x^2+z^2$ may vanish, in which case the orthogonal decomposition fails.
 Then $z = \pm ix$.
 Note that if $z=ix$:
 \begin{equation}
  K^{-1} M = \frac{1}{2}\begin{pmatrix}1& -i \\ 1&i\end{pmatrix} \begin{pmatrix} x&y\\ix&w\end{pmatrix} = \frac{1}{2}\begin{pmatrix} 2x & y-iw \\ 0 & y+iw \end{pmatrix},
 \end{equation}
 is upper triangular, and similarly with $KX$ if $z=-ix$.
 
 Both types of orthogonal decomposition fail only if $x^2+z^2=0$ and $y^2+w^2=0$ simultaneously.
 Again, write $z=\pm ix$.
 Then, by invertibility of $M$, $w=\mp iy$.
 Thus, letting $D=\left(\begin{smallmatrix}x&0\\0&y\end{smallmatrix}\right)$:
 \begin{equation}
  M = \begin{pmatrix} x&y\\\pm ix&\mp iy\end{pmatrix} = \begin{pmatrix} 1&1\\\pm i&\mp i\end{pmatrix} \begin{pmatrix} x&0\\0&y\end{pmatrix} = \begin{cases} KD &\text{if $\pm$ goes to $+$, or} \\ KXD & \text{if $\pm$ goes to $-$.} \end{cases}
 \end{equation}
 This completes the proof.
\end{proof}

Using the complex QR decomposition, we can now consider the solutions of $A^TA\doteq X$, where `$\doteq$' denotes equality up to scalar factor.

\begin{prop}
 The solutions of:
 \begin{equation}\label{eq:ATA-X}
  A^TA \doteq X
 \end{equation}
 are exactly those matrices $A$ satisfying $A=KD$ or $A=KXD$ for some invertible diagonal matrix $D$.
\end{prop}
\begin{proof}
 First, we check that matrices of the form $KD$ or $KXD$ for some invertible diagonal matrix $D$ satisfy \eqref{eq:ATA-X}.
 Indeed:
 \begin{equation}
  (KD)^TKD = D^T K^T K D = D \begin{pmatrix}1&i\\1&-i\end{pmatrix} \begin{pmatrix}1&1\\i&-i\end{pmatrix} D = 2 D X D = 2 xy X \doteq X,
 \end{equation}
 where $D = \left(\begin{smallmatrix}x&0\\0&y\end{smallmatrix}\right)$ with $x,y\in\CC\setminus\{0\}$.
 Similarly:
 \begin{equation}
  (KXD)^TKXD = D^T X^T K^T KXD = 2DX^3D = 2DXD = 2xyX \doteq X.
 \end{equation}
 This completes the first part of the proof.

 It remains to be shown that these are the only solutions of $A^TA\doteq X$.
 Assume, for the purposes of deriving a contradiction, that there is a solution that has an orthogonal QR decomposition.
 In particular, suppose $A=QR$ for some upper triangular matrix $R$.
 Then, by orthogonality of $Q$:
 \begin{equation}
  A^T A = R^TQ^TQR = R^T R = \begin{pmatrix} R_{00} & 0 \\ R_{01} & R_{11} \end{pmatrix} \begin{pmatrix} R_{00} & R_{01} \\ 0 & R_{11} \end{pmatrix} = \begin{pmatrix} R_{00}^2 & R_{00}R_{01} \\ R_{00}R_{01} & R_{01}^2+R_{11}^2 \end{pmatrix}.
 \end{equation}
 The only way for the top left component of this matrix to be zero, as required, is if $R_{00}$ is zero. 
 Yet, in that case, the top right and bottom left components of $A^T A$ are zero too, hence $A^T A$ cannot be invertible.
 That is a contradiction because any non-zero scalar multiple of $X$ is invertible.
 
 A similar argument applies if $R$ is lower triangular.
 
 Thus, all solutions of \eqref{eq:ATA-X} have to fall into the third case of Lemma \ref{lem:QR_decomposition}: i.e.\ all solutions must be of the form $A=KD$ or $A=KXD$.
\end{proof}